\documentclass[11pt]{amsart}
\usepackage{graphicx,amssymb,amsmath,amsthm}
\usepackage{enumerate}
\usepackage{dsfont}
\usepackage[colorlinks, citecolor=red]{hyperref}
\usepackage{comment,cite,color}
\usepackage{cite,color}
\usepackage{mathrsfs}
\usepackage{epsfig}
\usepackage{lscape}
\usepackage{subfigure}
\usepackage{epstopdf}
\usepackage{caption}
\usepackage{algorithm}

\usepackage{tgtermes}
\usepackage{lipsum}

\usepackage{bm}

\usepackage{algpseudocode}

\newcommand{\xkh}[1]{\left(#1\right)}
\newcommand{\dkh}[1]{\left\{#1\right\}}

\newcommand{\nj}[1]{\langle {#1} \rangle}

\newcommand{\norm}[1]{\|{#1}\|_2}
\newcommand{\normf}[1]{\|{#1}\|_F}
\newcommand{\normone}[1]{\|{#1}\|_1}
\newcommand{\norms}[1]{\|{#1}\|}
\newcommand{\abs}[1]{\left\lvert#1\right\rvert}

\newcommand{\A}{{\mathcal A}}
\renewcommand{\H}{{\mathcal H}^{d\times d}}
\newcommand{\E}{{\mathbb E}}

\newcommand{\PP}{{\mathbb P}}

\newcommand{\1}{{\mathds 1}}

\newcommand{\R}{{\mathbb R}}
\newcommand{\Rd}{{\mathbb R}^d}

\newcommand{\T}{\top}
\newcommand{\C}{{\mathbb C}}
\newcommand{\Cd}{{\mathbb C}^d}

\newcommand{\x}{{\widehat{\bm{x}}}}
\newcommand{\z}{{\widehat{\bm{z}}}}
\newcommand{\X}{{\widehat{X}}}

\newcommand{\vx}{{\bm x}}

\newcommand{\vu}{{\bm u}}
\newcommand{\vv}{{\bm v}}
\newcommand{\vz}{{\bm z}}
\newcommand{\vb}{{\bm b}}

\newcommand{\ta}{{\tilde{\bm{a}}}}

\newcommand{\va}{{\bm a}}
\newcommand{\vh}{{\bm h}}
\renewcommand\refname{Bibliography}

\newcommand{\vbar}{\bar{{\bm v}}}

\newcommand{\cN}{{\mathcal N}}

\newcommand{\rank}{{\rm rank}}

\newcommand{\supp}{{\rm supp}}

\renewcommand{\omega}{\eta}

\newcommand{\RNum}[1]{\uppercase\expandafter{\romannumeral #1\relax}}

\newtheorem{definition}{Definition}[section]
\newtheorem{corollary}[definition]{Corollary}

\newtheorem{theorem}[definition]{Theorem}
\newtheorem{lemma}[definition]{Lemma}

\newtheorem{remark}[definition]{Remark}

\newtheorem{example}[definition]{Example}

\date{}

\begin{document}
\baselineskip 19pt
\bibliographystyle{plain}
\title{Performance bounds of the intensity-based estimators for noisy phase retrieval}
\author{Meng Huang}
\address{Department of Mathematics, The Hong Kong University of Science and Technology,
Clear Water Bay, Kowloon, Hong Kong, China} \email{menghuang@ust.hk}

\author{Zhiqiang Xu}
\thanks{Zhiqiang Xu was supported by the National Science Fund for Distinguished Young Scholars grant (12025108), by Beijing Natural Science Foundation (Z180002) and  by NSFC grant (12021001, 11688101).  }
\address{LSEC, Inst.~Comp.~Math., Academy of
Mathematics and System Science,  Chinese Academy of Sciences, Beijing, 100091, China
\newline
School of Mathematical Sciences, University of Chinese Academy of Sciences, Beijing 100049, China}
\email{xuzq@lsec.cc.ac.cn}

\maketitle




\begin{abstract}
The aim of noisy  phase  retrieval is to  estimate a signal $\vx_0\in \Cd$ from $m$
noisy intensity measurements $b_j=\abs{\nj{\va_j,\vx_0}}^2+\eta_j, \; j=1,\ldots,m$,
where $\va_j \in \Cd$ are known measurement vectors and
$\eta=(\eta_1,\ldots,\eta_m)^\T \in \R^m$ is a  noise vector. A commonly used estimator for
 $\vx_0$ is to minimize the intensity-based  loss function, i.e.,
$\x:=\mbox{argmin}_{\vx\in \Cd} \sum_{j=1}^m \big(\abs{\nj{\va_j,\vx}}^2-b_j
\big)^2$. Although one has  developed many algorithms for solving  the
intensity-based estimator, there are very few results about its estimation
performance. In this paper, we focus on  the performance of the intensity-based
estimator and prove that the error bound satisfies $\min_{\theta\in
\R}\|\x-e^{i\theta}\vx_0\|_2 \lesssim \min\Big\{\frac{\sqrt{\|\eta\|_2}}{{m}^{1/4}},
\frac{\|\eta\|_2}{\norm{\vx_0}\cdot \sqrt{m}}\Big\}$ under the assumption of  $m
\gtrsim d$ and $\va_j \in \Cd, j=1,\ldots,m,$ being complex Gaussian random vectors.
We also show that the error bound is rate optimal when $m\gtrsim d\log m$.
  For the case where $\vx_0$ is an $s$-sparse signal, we present a similar result under the assumption of $m \gtrsim
s \log (ed/s)$. To the best of our knowledge, our results are the first theoretical
guarantees for the intensity-based estimator  and its sparse version. Our proofs employ
Mendelson's small-ball method which can deliver an effective lower bound on a
nonnegative empirical process.
\end{abstract}

\smallskip

\section{Introduction}
\subsection{Phase retrieval}
Assume that
\[
b_j:=\abs{\nj{\va_j,\vx_0}}^2+\eta_j, \quad  j=1,\ldots,m
\]
where $ \va_j \in \Cd$ are known measurement vectors and
$\eta:=(\eta_1,\ldots,\eta_m)^\T \in \R^m$ is a  noise vector. Throughout this paper,
we assume that the noise $\eta$ is a  fixed or random vector independent of measurement
vectors $ \va_j, \; j=1,\ldots,m$.

To estimate $\vx_0\in \C^d$ from $\vb:=(b_1,\ldots,b_m)^\T\in \R^m$ is referred to as
{\em phase retrieval}. Due to the physical limitations, optical sensors can record
only the modulus of Fraunhofer diffraction pattern while losing the phase
information, and hence phase retrieval has many applications in fields of physical
sciences and engineering, which includes X-ray crystallography
\cite{harrison1993phase,millane1990phase}, microscopy \cite{miao2008extending},
astronomy \cite{fienup1987phase}, coherent diffractive imaging
\cite{shechtman2015phase,gerchberg1972practical} and optics
\cite{walther1963question} etc. Despite its simple mathematical form, it has been
shown that to reconstruct a finite-dimensional discrete signal from its Fourier
transform magnitudes is generally {\em NP-complete} \cite{Sahinoglou}.

Based on the least squares criterion, one can employ the following {\em
intensity-based} empirical loss to estimate  $\vx_0$:
\begin{equation}\label{eq:mod1}
\min_{\vx\in \C^d} \quad \sum_{j=1}^m \xkh{\abs{\nj{\va_j,\vx}}^2-b_j}^2.
\end{equation}
For the case where $\vx_0$ is sparse,  the following Lasso-type program can be
employed  to estimate $\vx_0$:
\begin{equation}\label{eq:mod2sparse}
\min_{\vx\in \C^d} \quad \sum_{j=1}^m \xkh{\abs{\nj{\va_j,\vx}}^2-b_j}^2 \quad \mbox{s.t.} \quad \normone{\vx}\le R,
\end{equation}
where  $R$ is a parameter which specifies a desired sparsity level of the solution.
An  advantage of the intensity-based estimator (\ref{eq:mod1}) is that the objective
function is differentiable based on the Wirtinger derivatives. One can therefore try to find its minimum with high-order algorithms
 such as trust-region and Gauss-Newton methods. Though the
objective functions are non-convex,  the strategy of spectral initialization plus
local gradient descent can be adopted to solve (\ref{eq:mod1}) and
(\ref{eq:mod2sparse}) efficiently under Gaussian random measurements.  For instance,
it has been proved that when $m\gtrsim d$ and $\norms{\eta}_{\infty}\lesssim
\norm{\vx_0}^2 $, with high probability the truncated spectral method given in
\cite{TWF} can return an initial guess $\vz_0$ which is close to the
target signal in the real case, namely, $\norm{\vz_0-\vx_0} \le \delta\norm{\vx_0}$ for any fixed relative error tolerance $\delta$. With this in place, the update rules
such as Wirtinger Flow \cite{WF}, Trust-Region \cite{turstregion} and Gauss-Newton
\cite{Gaoxu}  methods  could find a global solution to (\ref{eq:mod1}) at least in the noiseless case.

In the noiseless case, i.e., $ b_j=\abs{\nj{\va_j,\vx_0}}^2, j=1,\ldots,m$, the
solution to (\ref{eq:mod1}) is exactly  $\vx_0$ (up to a unimodular constant) if $m\geq 4d-4$
and $\va_j,j=1,\ldots,m,$ are generic vectors in $\C^d$ \cite{CEHV,WX19}. However,
 one still does not know the distance between the solution to
(\ref{eq:mod1}) and the true signal $\vx_0$ in the noisy case.
 The aim of this paper is to study the performance of (\ref{eq:mod1}) and
(\ref{eq:mod2sparse}) from the theoretical viewpoint.

\subsection{Algorithms for phase retrieval }
For the last two decades, many algorithms  have been designed for phase retrieval, especially in the noiseless case,
which falls into two categories: convex methods and non-convex ones.

The convex methods rely on the ``matrix-lifting''  technique which lifts the
quadratic system to a linear rank-one positive semi-definite program.
More specifically, a rank one matrix $X=\vx \vx^*$ is introduced to linearize the quadratic constrains
and then a nuclear norm minimization is adopted as a convex surrogate of the rank constraint.
Such methods include PhaseLift \cite{phaselift,Phaseliftn}, PhaseCut
\cite{Waldspurger2015} etc. Although the convex methods have good theoretical
guarantees,  they require to solve a semi-definite program in the ``lifted'' space
$\C^{d^2}$ rather than $\C^d$, where $d$ is the dimension of signals. Thus the memory
requirements and computational complexity become quite high, which makes it
prohibitive for large-scale problems in practical applications.

The non-convex methods operate directly on the original space, which achieves
significantly improved computational performance. The oldest non-convex algorithms
for phase retrieval are based on alternating projection including Gerchberg-Saxton
\cite{gerchberg1972practical} and Fineup \cite{ER3}, but lack of theoretical guarantees. The
first non-convex algorithm with theoretical guarantees was given by  Netrapalli et al
who showed that the  AltMinPhase \cite{AltMin} algorithm converges  linearly to the
true solution up to a global phase with $O(d \log^3 d)$ resampling Gaussian random
measurements. In \cite{WF}, Cand\`es,  Li and Soltanolkotabi developed the Wirtinger Flow (WF)  to
solve (\ref{eq:mod1}) and proved WF algorithm can achieve the linear convergence with
$O(d \log d)$  Gaussian random measurements. Lately, Chen and Cand\`es improved the
result to $O(d )$ Gaussian random measurements by Truncated Wirtinger Flow (TWF)
\cite{TWF}. In \cite{Gaoxu}, Gao and Xu proposed a Gauss-Newton  algorithm to solve
(\ref{eq:mod1}) and proved the Gauss-Newton method can achieve quadratic convergence
for the real-valued signals  with $O(d \log d)$ Gaussian random measurements. In
\cite{turstregion}, Sun, Qu, and Wright proved that, for $O(d \log^3 d)$ Gaussian random measurements,
the objective function of (\ref{eq:mod1}) has a benign geometric landscape: (1) all local minimizers are global;
and (2) the objective function has a negative curvature around each saddle point\footnote{We do not differentiate between saddle points and local maximizers.}. They also developed the Trust-Region method to find a global
solution.

Another alternative approach for phase  retrieval is to solve the following {\em
amplitude-based} empirical loss:
\begin{equation}\label{eq:mod2}
\min_{\vx\in \C^d}\,\,\sum_{j=1}^m \xkh{\abs{\nj{\va_j,\vx}}-\psi_j}^2,
\end{equation}
where $\psi_j:=\sqrt{b_j}, \; j=1,\ldots, m$.  For Gaussian random measurements,
through an appropriate initialization, many algorithms can be used to solve
(\ref{eq:mod2}) successfully such as Truncated Amplitude Flow (TAF) \cite{TAF},
Reshaped Wirtinger Flow (RWF) \cite{RWF} and Perturbed Amplitude Flow (PAF)
\cite{PAF}.  It has been proved that TAF, RWF and PAF algorithms converge linearly to
the true solution up to a global phase under $O(d )$ Gaussian random measurements.

For sparse phase retrieval, a standard $\ell_1$ relaxation technique leads
to the corresponding sparse intensity-based estimator (\ref{eq:mod2sparse}).
 It has
 been shown that when the noises are independent centered sub-exponential random
variables with maximum sub-exponential norm $\sigma$ and $m\gtrsim \alpha_\sigma^2
s^2 \log(md)$, in the real case, the initialization procedure given in
\cite{cai2016optimal} can return an initial guess $\vz_0$ which satisfies
$\supp(\vz_0) \subset \supp(\vx_0)$ and  is very close to the target signal with high
probability, where $s$ is the sparsity level and $\alpha_\sigma$ is a constant
related to $\sigma$. Next, the projection gradient method \cite{Duchi08} could be
used to find a global minimizer to (\ref{eq:mod2sparse}). Such two-step procedure has
also been used in various signal processing and machine learning problems, such as
Blind Deconvolution \cite{Ling2019Regularized}, matrix completion
\cite{Sun2016Guaranteed} and sparse recovery \cite{Qu2016Finding}.

We refer the reader to survey papers \cite{SparseTAF, Hand2016,Wu2020} for accounts of recent
developments in the algorithms of sparse phase retrieval. The theoretical results
concerning the injectivity of sparse phase retrieval can be found in \cite{WX14,
Iwen2015Robust}.

\subsection{Related work}
\subsubsection{PhaseLift}
We first introduce the  estimation performance of PhaseLift for noisy phase
retrieval. In \cite{Phaseliftn}, Cand\`es and Li suggest using the following empirical  loss
to estimate $\vx_0$:
\begin{equation} \label{mo:phaselift}
\min_{X \in \C^{d\times d}} \quad \sum_{j=1}^m \abs{\va_j^* X \va_j-b_j} \quad \mbox{s.t.} \quad X \succeq 0.
\end{equation}
They prove that  the solution $\X$ to (\ref{mo:phaselift}) obeys
\[
\normf{\X-\vx_0\vx_0^*} \lesssim \frac{\normone{\eta}}{m}
\]
with high probability provided  $m\gtrsim  d$ and $\va_j \in \C^d, j=1,\ldots,m,$ are complex Gaussian
random vectors.
Though (\ref{mo:phaselift}) is a convex optimization problem, one needs to solve it in a ``lifted"
space $\C^{d^2}$. The computational cost typically far exceeds the order of $d^3$, which is not suitable for
large-dimensional data. However, for the intensity-based estimator, one just needs to operate on
the original space $\C^d$ rather than lifting the problem into higher dimensions.


\subsubsection{The amplitude-based estimator}
As shown before, the amplitude-based empirical loss (\ref{eq:mod2}) is an alternative estimator for phase retrieval.
In \cite{Mengxu}, Huang and Xu studied the estimation performance of the
amplitude-based estimator (\ref{eq:mod2}) for {\em real-valued } signals. They prove that the
solution $\x$ to (\ref{eq:mod2}) satisfies
\[
\min \{\norm{\x+\vx_0}, \norm{\x-\vx_0}\}\lesssim \frac{\norm{\eta}}{\sqrt{m}}
\]
with high probability provided  $m \gtrsim d$ and $\va_j \in \R^d,j=1,\ldots,m,$ are Gaussian
random vectors.
 They also prove that the reconstruction error $\norm{\eta}/\sqrt{m}$ is sharp.
Furthermore, in \cite{Mengxu}, Huang and Xu  consider the following constrained
nonlinear Lasso to estimate  $s$-sparse signals $\vx_0$:
\begin{equation} \label{mo:conslasso}
\min_{\vx\in \R^d}\,\,\sum_{j=1}^m \xkh{\abs{\nj{\va_j,\vx}}-\psi_j}^2 \quad \mbox{s.t.} \quad \normone{\vx}\le R.
\end{equation}
They show  that    any global solution $\x$ to (\ref{mo:conslasso}) with $R:=
\normone{\vx_0}$ obeys
\[
\min \{\norm{\x+\vx_0}, \norm{\x-\vx_0}\}\lesssim \frac{\norm{\eta}}{\sqrt{m}}
\]
with high probability provided $m\gtrsim  s \log (ed/s)$ and $\va_j \in \R^d, j=1,\ldots,m,$
are Gaussian random vectors.

The results from \cite{Mengxu} hold only for real-valued signals and it seems highly
nontrivial to extend them to the complex case. However, the results in this paper
hold for both real-valued and complex-valued signals. Furthermore, the intensity-based estimator considered in this paper
 is differentiable comparing with amplitude-based estimator,  which admits high-order algorithms.

\subsubsection{Poisson log-likelihood estimator}
In \cite{TWF}, Chen and Cand\`es consider the Poisson log-likelihood function
\[
\ell_j(\vx,b_j)=b_j \log (|\va_j^\T \vx|^2)-|\va_j^\T \vx|^2
\]
and use
\begin{equation} \label{es:poisson}
\min_{\vx \in \R^d }\quad - \sum_{j=1}^m \ell_j(\vx,b_j)
\end{equation}
to estimate real-valued signals $\vx_0$. They establish stability estimates using Truncated Wirtinger Flow approach and show that if $\norms{\eta}_\infty \lesssim \norm{\vx_0}^2$ then the solution $\x$ to (\ref{es:poisson}) satisfies
\[
\mbox{dist}~(\x,\vx_0) \lesssim  \frac{\norm{\eta}}{\norms{\vx_0}\sqrt{m}}
\]
with high probability provided  $m\gtrsim  d$ and $\va_j \in \R^d, j=1,\ldots,m,$ are Gaussian random vectors. Furthermore, a lower bound on the minimax estimation error is also derived under the real-valued Poisson noise model, namely, with high probability
\begin{equation}\label{lower:poisson}
\mbox{inf}_{\x} \quad \mbox{sup}_{\vx_0 \in \gamma(K)}\quad  \E~ \mbox{dist}~(\x,\vx_0) \gtrsim \sqrt{\frac{d}{m}},
\end{equation}
where
$
\gamma(K):=\dkh{\vx_0 \in \R^d: \norm{\vx_0} \in (1\pm 0.1) K}
$
for any $K\ge \log^{1.5} m$.

The result  (\ref{lower:poisson}) presents a lower bound on the expectation of the  minimax
estimation error under Poisson noise structure, whereas the result in Theorem
\ref{th:lowbound} gives a tail bound with a reasonably large measurements and holds
for a wide range of noises. Moreover, all the results in \cite{TWF} are for
real-valued signals, while ours hold for complex-value ones.

\subsection{Our contributions} \label{sec:contrib}
As stated earlier, the two-step strategy of spectral initialization plus local gradient descent can be use
to solve the intensity-based estimator (\ref{eq:mod1}). To our knowledge, there is no result concerning the
reconstruction error of (\ref{eq:mod1}) for noisy phase retrieval from the
theoretical viewpoint. The goal of this paper is to study the estimation performance
of the intensity-based estimator (\ref{eq:mod1}) and its sparse version
(\ref{eq:mod2sparse}).

Our first  result shows that  the estimation error is quite small and bounded by the
average noise per measurement, as stated below. We emphasize that this theorem does
not assume any particular structure on the noise $\eta$.

\begin{theorem} \label{result1}
Suppose that the measurements  $\va_j \sim  1/\sqrt{2}\cdot
\cN(0,I_d)+i/\sqrt{2}\cdot \cN(0,I_d) $ are i.i.d. complex Gaussian random vectors
and the measurement number  $m \gtrsim d$. Then the following holds with probability
at least $1-\exp(-c m)$: For all $\vx_0\in \Cd$, the solution $\x \in \Cd$  to
(\ref{eq:mod1}) with $b_j=\abs{\nj{\va_j,\vx_0}}^2+\eta_j, \;  j=1,\ldots,m$,
satisfies
\begin{equation*}\label{eq:thup}
  \min_{\theta\in [0,2\pi)} \norm{\x-e^{i\theta}\vx_0}\,\,\le\,\,
  C\min \dkh{\frac{\sqrt{\|\eta\|_2}}{{m}^{1/4}}, \frac{\|\eta\|_2}{\norm{\vx_0}\cdot \sqrt{m}}}.
\end{equation*}
Here, $C$ and $c$ are positive absolute constants.
\end{theorem}

According to Theorem \ref{result1}, the following holds with high probability:
\begin{equation*}
   \min_{\theta\in [0,2\pi)} \norm{\x-e^{i\theta}\vx_0}
\lesssim\,\,
  \min \dkh{\frac{\sqrt{\|\eta\|_2}}{{m}^{1/4}}, \frac{\|\eta\|_2}{\norm{\vx_0}\cdot \sqrt{m}}}
   = \left\{
   \begin{array}{ll}
   \frac{\sqrt{\|\eta\|_2}}{{m}^{1/4}}  \quad & \mbox{if} \quad \|\vx_0\|_2< \frac{\sqrt{\|\eta\|_2}}{{m}^{1/4}}\\
    \frac{\|\eta\|_2}{\norm{\vx_0}\cdot \sqrt{m}} \quad & \mbox{if} \quad \|\vx_0\|_2\ge \frac{\sqrt{\|\eta\|_2}}{{m}^{1/4}}
   \end{array}
   \right..
\end{equation*}
 In practical applications, the signals of interest $\vx_0$ usually fall in the second regime where $\|\vx_0\|_2\ge \frac{\sqrt{\|\eta\|_2}}{{m}^{1/4}}$.


The next theorem presents a lower bound for the estimation error  for any fixed
$\vx_0$, under the assumption of   $m\gtrsim d\log m$  and noise with the structures
$\norm{\eta} \asymp \sqrt{m}$, $\norms{\eta}_{\infty}\lesssim \log m$, $~| \sum_{j=1}^m \eta_j|  \lesssim m $ and $\big| \sum_{j=1}^m \eta_j\big| \gtrsim \sqrt{m}\norm{\eta} \gtrsim d
\norms{\eta}_{\infty} $.
 The result shows that the estimator (\ref{eq:mod1}) is rate optimal for
some  $\eta$ provided $m\gtrsim d\log d$.


\begin{theorem}\label{th:lowbound}
Suppose that the measurements $\va_j \sim  1/\sqrt{2}\cdot \cN(0,I_d)+i/\sqrt{2}\cdot
\cN(0,I_d) $ are i.i.d. complex Gaussian random vectors and the measurement number $m
\gtrsim d \log m$. Assume that $\eta \in \R^m$ is a noise vector  satisfying
$\norm{\eta} \asymp \sqrt{m}$,
 $\norms{\eta}_{\infty}\lesssim \log m$, $~| \sum_{j=1}^m \eta_j|  \lesssim m $
 and $d \norms{\eta}_{\infty} \lesssim
 \sqrt{m}\norm{\eta} \lesssim \big| \sum_{j=1}^m \eta_j\big|$. For  any fixed $\vx_0\in \Cd$ satisfying $\norm{\vx_0} \ge
2C\norm{\eta}/\sqrt{m}$ the following holds with probability at least $1-c'\exp(-c''
d)-c'''m^{-1}$: any solution $\x \in \Cd$ to (\ref{eq:mod1}) with
$b_j=\abs{\nj{\va_j,\vx_0}}^2+\eta_j, \; j=1,\ldots,m$, satisfies
\begin{equation*}
  \min_{\theta\in [0,2\pi)} \norm{\x-e^{i\theta}\vx_0}\,\,\gtrsim\,\,
  \min \dkh{\frac{\sqrt{\|\eta\|_2}}{{m}^{1/4}}, \frac{\|\eta\|_2}{\norm{\vx_0}\cdot \sqrt{m}}}.
\end{equation*}
Here, $ c', c'', c'''$ are universal positive constants and $C$ is the constant in Theorem \ref{result1}.
\end{theorem}


\begin{remark}
In Theorem \ref{th:lowbound}, we require $\eta$ satisfies   the conditions
$\norm{\eta} \asymp \sqrt{m}$, $\norms{\eta}_{\infty}\lesssim \log m$, $~|
\sum_{j=1}^m \eta_j| \lesssim m $ and $d \norms{\eta}_{\infty} \lesssim
 \sqrt{m}\norm{\eta} \lesssim \big| \sum_{j=1}^m \eta_j\big|$. In fact, there exist many noises
satisfying them. For instance, if $\eta_j$ are generated independently according to the
biased Gaussian distribution, i.e., $\eta_j \sim N(\mu,\sigma^2)$ for some non-zero
constant  $ \mu$, then the noise vector $\eta$ satisfies those conditions with high
probability.
\end{remark}



%

We next turn to the phase retrieval for sparse signals. This is motivated by the
signal $\vx_0 \in \C^d$ admitting a sparse representation under some linear
transformation in many applications. Without loss of generality, we  assume that
$\vx_0 \in \Cd$ is an $s$-sparse vector and wish to estimate $\vx_0$ from
$\vb=(b_1,\ldots,b_m)^\T$ by solving
\begin{equation}\label{quartic model sparse}
\min_{\vx\in \C^d} \quad \sum_{j=1}^m \xkh{\abs{\nj{\va_j,\vx}}^2-b_j}^2 \quad \mbox{s.t.} \quad \normone{\vx}\le R.
\end{equation}
The estimation performance of (\ref{quartic model sparse}) is stated as follows.

\begin{theorem} \label{result sparse}
Suppose that the measurements $\va_j \sim  1/\sqrt{2}\cdot \cN(0,I_d)+i/\sqrt{2}\cdot
\cN(0,I_d) $ are i.i.d. complex Gaussian random vectors and the measurement number $m
\gtrsim s \log (ed/s)$. Then the following holds with probability at least
$1-5\exp(-c m)$ where $c$ is a constant: for any $s$-sparse vector $\vx_0\in \Cd$,
\begin{equation*}
  \min_{\theta\in [0,2\pi)} \norm{\x-e^{i\theta}\vx_0}  \; \lesssim \; \min \dkh{\frac{\sqrt{\|\eta\|_2}}{{m}^{1/4}}, \frac{\|\eta\|_2}{\norm{\vx_0}\cdot \sqrt{m}}},
\end{equation*}
where $\x \in \Cd$ is a solution to (\ref{quartic model sparse}) with parameter
$R:=\normone{\vx_0}$.
\end{theorem}

\begin{remark}
In \cite{cai2016optimal}, the authors establish the estimation error using the
Thresholded Wirtinger Flow approach under the centered sub-exponential noise for the
real-valued signals. In short, they show that, with probability at least
$1-47/m-10/e^s$, the estimator $\x$ given by the Thresholded Wirtinger Flow algorithm
obeys
$
\min\dkh{\x-\vx_0, \x+\vx_0} \lesssim \frac{\sigma}{\norm{\vx_0}} \sqrt{\frac{s\log
d}{m}}
$
provided $m\ge O(s^2 \log(md))$, where $\sigma:=\max_{1\le j\le m}
\norms{\eta_j}_{\psi_1}$. Although the estimation error is slightly better than the
upper bound given in Theorem \ref{result sparse}, however, our result holds for any
noise structure and complex-valued signals.  Moreover, the probability of failure in
Theorem \ref{result sparse}  is exponentially small in the number of measurements.
\end{remark}

\subsection{Numerical Experiments}
In this subsection, we report some numerical experiments to verify  that
the global solutions to (\ref{eq:mod1}) and (\ref{eq:mod2sparse}) can
be obtained efficiently and the results given in Subsection \ref{sec:contrib} are rate optimal.
In our experiments, the target signal $\vx_0$ and the measurement vectors
$\va_1,\ldots,\va_m$ are independent standard complex Gaussian random vectors,
whereas the noise vector $\eta \in \R^m $ is a real Gaussian random vector with
entries $\eta_j \sim  N(1,1)$.

\begin{example}
In this example, we verify the estimation error presented in Theorem \ref{result1} is
rate optimal.  We consider the case where $d=500$ and vary $m$ within the range
$[4d,50d]$. To solve the estimator (\ref{eq:mod1}), we use the truncated spectral
method proposed in \cite{TWF} to obtain a good initial guess and then refine it by
Wirtinger Flow \cite{WF}. Figure \ref{figure:1} depicts the ratio $\rho_m$ against
the number of measurements $m$,  when averaged over $100$ times independent trials.
Here, the ratio $\rho_m$ is defined as
\begin{equation} \label{eq:rhom}
\rho_m:= \frac{\text{dist}(\x,\vx_0)}{\norm{\eta}/(\norm{\vx_0}\cdot \sqrt{m})}.
\end{equation}
Numerical results show that $\rho_m$ tends to be a constant around $0.37$, which
verifies the estimation error $\frac{\norm{\eta}}{\norm{\vx_0}\cdot \sqrt{m}}$
presented in Theorem \ref{result1} is rate optimal.
\end{example}

\begin{figure}[H]
\centering
\includegraphics[width=0.4\textwidth]{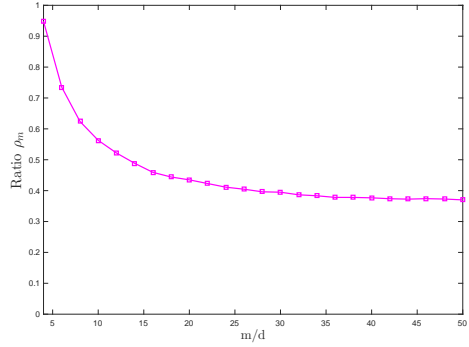}
\caption{The ratio $\rho_m$ versus the number of measurements $m$ under Gaussian noises with $d=500$.}
\label{figure:1}
\end{figure}

\begin{example}
The purpose of this numerical experiment is to verify  the estimation bound given in
Theorem \ref{result sparse} is  rate optimal when $m=O(s \log (ed/s))$. We choose
$d=1000$ and take the sparsity level $s=100$. The support of $\vx_0$ is uniformly
distributed at random.  The non-zero entries of $\vx_0$ are chosen randomly according to a
standard normal distribution. We vary $m$ between $\lceil 6s\log (ed/s)\rceil$ and
$\lceil 20s\log (ed/s)\rceil $. For each fixed $m$, we run $100$ times trials and
calculate the average ratio $\rho_m$ defined in \eqref{eq:rhom}. The constrained
optimization problem (\ref{quartic model sparse}) is solved by combining the
initialization method introduced in \cite{cai2016optimal} and the projection gradient
descent onto the $\ell_1$- ball \cite{Duchi08}. The result is plotted in Figure
\ref{figure:2}. We can see that $\rho_m$ tends to be a constant around $0.72$, which
verifies the estimation error $\frac{\norm{\eta}}{\norm{\vx_0}\cdot \sqrt{m}}$
presented in Theorem \ref{result sparse} is rate optimal for sparse signals.
\end{example}

\begin{figure}[H]
\centering
\includegraphics[width=0.4\textwidth]{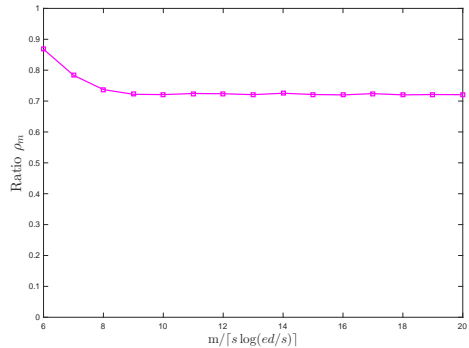}
\caption{The ratio $\rho_m$ versus the number of measurements $m$ for $s$-sparse signals with $d=1000$, $s=100$.}
\label{figure:2}
\end{figure}

\subsection{Notations}
Throughout this paper, we assume the measurements $\va_j\in \Cd, \; j=1,\ldots,m $
are i.i.d. complex Gaussian random vectors. Here we say  $\va\in \Cd$ is a complex Gaussian
random vector if $\va \sim 1/\sqrt{2}\cdot \cN(0,I_d)+i/\sqrt{2}\cdot \cN(0,I_d)$.
We write $\vz \in \mathbb{S}_{\C}^{d-1}$ if $\vz \in \C^d$ and $\norm{\vz}=1$.
We use the notations $\norm{\cdot}$ and $\norms{\cdot}_*$ to denote the operator norm and
nuclear norm of a matrix, respectively. For any $A,B\in \R$, we use $ A \lesssim B$
to denote $A\le C_0 B$ where $C_0\in \R_+$ is an  absolute constant.  The notion
$\gtrsim$ can be defined similarly. Moreover, $A \asymp B$ means that there exist constants $C_1,C_2>0$ such that $C_1 A \le B \le C_2 A$.
 In this paper, we use  $C,c$ and the subscript
(superscript)   form of them to denote universal constants whose values vary with the
context.

\subsection{Organization}
The paper is organized as follows. In Section 2, after introducing some definitions,
we study the recovery of low-rank matrices from rank-one measurements, which plays a
key role in the proofs of main results. We also believe that the results in Section 2
are of independent interest. Combining the  Mendelson's small-ball method and the
results in Section 2, we present the proofs of Theorem \ref{result1} and Theorem
\ref{th:lowbound} in Section 3.  The proof of Theorem \ref{result sparse}  is given
in Section 4. A brief discussion is presented in Section 5.  Appendix collects the
technical lemmas needed in the proof.

\section{ The recovery of low-rank matrices from rank-one measurements }

 For convenience, we let $\A: \H(\C) \to
\R^m$ be a linear map  which is defined as
\begin{equation} \label{eq:mapA}
\A(X):=\xkh{\va_1^* X\va_1, \va_2^* X\va_2 ,\ldots, \va_m^* X\va_m }^\T,
\end{equation}
where $\H(\C):=\{X\in \C^{d\times d}: X^*=X\}$.
Its dual operator $\A^*: \R^m\to  \H(\C)$ is given by
\begin{equation}\label{eq:dualo}
\A^*(\vz)=\sum_{j=1}^m z_j \va_j\va_j^*.
\end{equation}
In this section, we focus on the following minimization problem:
\begin{equation} \label{mod:determ}
\min_{X\in \H(\C)}  \norm{\A(X)-\vb}^2 \quad \mbox{s.t.}\quad  \rank{(X)} \le r.
\end{equation}
Set $\H_r(\C):=\{X\in \C^{d\times d}: X^*=X, {\rm rank}(X)\leq r\}$.
 A simple observation is  that $\x$ is a solution to
(\ref{eq:mod1}) if and only if $\X:=\x\x^*$ is a solution to (\ref{mod:determ}) with
$r=1$. Hence, (\ref{mod:determ}) can be regarded as a lifted version of
(\ref{eq:mod1}).
To prove the main results of this paper, we first characterize the estimation
performance of  (\ref{mod:determ}).

\subsection{The performance of (\ref{mod:determ})}
The main result of this section is  Corollary  \ref{th:liftresult} which presents the
estimation performance of (\ref{mod:determ}). We believe some results in this section are also of independent interest in the area  of low-rank matrix recovery from rank-one measurements
\cite{HX,rank1,cai2015rop,chen2015exact}.

We first introduce the definition of Lower Restricted Isometry Property ( see, e.g., \cite{bourrier2014f,keriven2018instance}).
\begin{definition}\label{de:rip}[Lower Restricted Isometry Property]
A linear map $\A: \H(\C) \to \R^m$ is said to have the  {\em Lower Restricted
Isometry Property} (LRIP) condition of order $r$ and constant $c_0$ if the following
holds
\[
\norm{\A(X)} /\sqrt{m} \ge c_0\normf{X}
\]
for all non-zero matrices $X\in \H_r(\C) $.
\end{definition}

With the LRIP condition in place, we can demonstrate that the optimization
(\ref{mod:determ}) is stable, as stated in the following theorem.
\begin{theorem} \label{th:detstab}
Suppose $\A$ satisfies the LRIP condition with order $2r$ and constant $c_0>0$, then
the solution $\X$ to (\ref{mod:determ})  satisfies
\[
 \normf{\X-X_0}   \lesssim  \frac{\norm{\eta}}{\sqrt{m}}
\]
for all matrices $X_0\in \H_r(\C)$ and $\vb=\A(X_0 )+\eta$  with the noise vector
$\eta \in \R^m$.
\end{theorem}
\begin{proof}
Since $\X$ is the global solution to (\ref{mod:determ}) and $X_0$ is a feasible
point, we have
\[
\norm{\A(\X)-\vb} \le \norm{\A(X_0)-\vb}.
\]
Noting $\vb=\A(X_0 )+\eta$, we obtain that
\[
\norm{\A(H)-\eta}\le \norm{\eta},
\]
where $H=\X-X_0 \in \H_{2r}(\C)$.  Since $\A$ satisfies the LRIP condition, we have
\[
c_0 \sqrt{m} \normf{H} \le \norm{\A(H)} \le \norm{\A(H)-\eta}+ \norm{\eta} \le 2 \norm{\eta}.
\]
Consequently,
\[
 \normf{H} \lesssim \norm{\eta}/\sqrt{m}.
\]
We arrive at the conclusion.
\end{proof}

The next result shows that $\A$ satisfies  the LRIP condition with high probability
provided $\va_j, j=1,\ldots,m$, are i.i.d. complex Gaussian random vectors. We
postpone its proof  to the end of this section.

\begin{theorem} \label{le:ORIP}
 Suppose that  $\va_j\in \C^d ,j=1,\ldots,m, $ are i.i.d. complex Gaussian random vectors and  $t, r \in {\mathbb Z}_{\geq 1} $ satisfy $t\cdot r<d$.
If $m \gtrsim tdr$
   then with probability at least $1-\exp(-c m)$, the linear map $\A$ defined
   in \eqref{eq:mapA} satisfies LRIP condition of order $t\cdot r$ and constant $c_0$,
   where  $c, c_0>0$ are constants independent of $d,r$ and $t$.
 \end{theorem}

As a direct consequence of Theorem \ref{th:detstab}  and Theorem \ref{le:ORIP}, the
estimation performance of optimization (\ref{mod:determ}) is given below.

\begin{corollary}\label{th:liftresult}
Suppose that  $\va_j\in \C^d, j=1,\ldots,m,$ are i.i.d. complex Gaussian random
vectors. If $m \gtrsim dr$, then the following  holds  with probability at least
$1-\exp(-c m)$: for any $X_0\in \H_r(\C)$, the solution $\X$ to (\ref{mod:determ})
with noisy measurements  $b_j=\va_j^* X_0 \va_j +\eta_j, j=1,\ldots,m$, satisfies
\begin{equation*}
  \normf{\X-X_0}\,\,   \lesssim\,\, \frac{\norm{\eta}}{\sqrt{m}},
\end{equation*}
where $\eta:=(\eta_1,\ldots,\eta_m)^\T \in \R^m$ is a noise vector and $c$ is a positive constant.
\end{corollary}
\begin{proof}
Taking $t=2$ in Theorem \ref{le:ORIP}, it then follows that with probability at least $1-\exp(-c m)$, the linear map
$\A$ defined in \eqref{eq:mapA} satisfies LRIP condition of order $2r$ and constant $c_0$ for
some constants $c,c_0>0$. Combining with Theorem \ref{th:detstab}, we complete the
proof.

\end{proof}

\subsection{Proof of Theorem \ref{le:ORIP}}

In this subsection, we will establish  the LRIP condition of $\A$. Before proceeding,
we gather some lemmas which are useful in our arguments.
\subsubsection{Lemmas}
The Mendelson's  {\em small-ball method} (see \cite{tropp2015convex}) plays a key role in
our proof, which is a  strategy to establish a lower bound for $\inf_{\vx\in E}
\sum_{j=1}^m \abs{\nj{\vx,\phi_j}}^2$ where $\phi_j\in \Rd$ are independent random
vectors and $E$ is a subset of $\Rd$.
\begin{lemma}\cite[Proposition 5.1]{tropp2015convex} \label{le:lower}
Fix $E\subset \Rd$ and let $\phi_1,\ldots,\phi_m$ be independent copies of a random
vector $\phi $ in $\Rd$. For any $\xi\ge 0$, set
\[
Q_\xi(E,\phi):=\inf_{\vu\in E} \PP\dkh{\abs{\nj{\phi,\vu}}\ge \xi}
\]
and
\[
W_m(E,\phi):=\E \sup_{\vu\in E} \nj{\vh,\vu} \quad \mbox{where} \quad \vh:=\frac{1}{\sqrt{m}} \sum_{j=1}^m \epsilon_j \phi_j
\]
where $\epsilon_1,\ldots,\epsilon_m$ are independent Rademacher random variables.
Then for any $\xi>0$ and $t>0$ the following holds with probability at least
$1-\exp(-t^2/2)$:
\[
\inf_{\vu\in E} \xkh{\sum_{j=1}^m \abs{\nj{\phi_j,\vu}}^2}^{1/2}\,\, \ge\,\, \xi\sqrt{m}Q_{2\xi}(E,\phi)-2W_m(E,\phi)-\xi t.
\]
\end{lemma}

The following lemma is a consequence of  the  classical Paley-Zygmund inequality
(e.g., \cite{kahane1993some,de2012decoupling}).
\begin{lemma}\cite[Lemma 7.16]{foucart2017mathematical} \label{le:palzyg}
If a nonnegative random variable Z has finite second moment, then
\[
\PP(Z>t)\ge \frac{(\E Z-t)^2}{\E Z^2}, \quad 0\le t\le \E Z.
\]
\end{lemma}

In addition, we also need the following lemma which presents an upper bound for the spectral norm of $\A^*(\epsilon)$ for a fixed independent Rademacher random vector $\epsilon \in \R^m$.

\begin{lemma} \label{le:simiA}
Suppose that $\va_j \in\Cd, j=1,\ldots,m $, are i.i.d. complex Gaussian random vectors and
$\epsilon_1,\ldots,\epsilon_m$ are independent Rademacher random variables. If $m
\gtrsim  d$, then $\E_{\epsilon,\A} \norm{\A^*(\epsilon)} \lesssim \sqrt{md}$, namely,
\[
\E ~ \frac{1}{\sqrt{m}}\norm{\sum_{j=1}^m \epsilon_j \va_j\va_j^*}
    \lesssim \sqrt{d}.
 \]
where $\epsilon:=(\epsilon_1,\ldots,\epsilon_m)^\T$.
\end{lemma}


\begin{proof}
We assume that $\mathcal{N}$
 is a $1/4$-net  of the complex unit sphere ${\mathbb S}^{d-1}\subset {\mathbb C}^d$. It then follows from \cite[Lemma 4.4.3]{Vershynin2018} that
\begin{equation} \label{eq:operate}
 \norm{\sum_{j=1}^m \epsilon_j \va_j\va_j^*} \le 2 \max_{\vx \in \mathcal{N}} \Big| \sum_{j=1}^m \epsilon_j \abs{\va_j^*\vx}^2\Big|.
\end{equation}
 For any fixed $\vx \in {\mathbb S}^{d-1}$,  the terms $\epsilon_j |\va_j^*\vx|^2, j=1,\ldots,m$ are
independent centered sub-exponential random variables with the sub-exponential norm
being a constant.  Using the Bernstein's inequality \cite[Theorem
2.8.1]{Vershynin2018}, we obtain that, for any $t\ge 0$, it holds
\begin{eqnarray*}
\PP\dkh{ \Big| \sum_{j=1}^m \epsilon_j \abs{\va_j^*\vx}^2\Big| \ge \frac{C\sqrt{md}+t\sqrt{m}}{2} } &\le & 2\exp\xkh{-c \min\Big(\frac{\lambda^2}{m}, \lambda\Big)}.
\end{eqnarray*}
where $\lambda:=C\sqrt{md}+t \sqrt{m}$ and $C\ge 1$ is a constant to be chosen later. Noting that
\[
\lambda^2/m=C^2 d + 2Ct\sqrt{d}+t^2\ge Cd+t \qquad \mbox{and} \qquad \lambda \ge Cd+t
\]
for any $m\gtrsim d$, we have
\[
\PP\dkh{ \Big| \sum_{j=1}^m \epsilon_j \abs{\va_j^*\vx}^2\Big| \ge \frac{C\sqrt{md}+t\sqrt{m}}{2} }  \le 2\exp\xkh{-c(Cd+t)}.
\]
Recall that $|\mathcal{N}|\le 9^{2d}$. Taking the constant $C$ such that
$C\cdot c\ge 2\ln 9$, we obtain
\begin{equation}\label{eq:Pup}
\PP\dkh{\max_{\vx \in \mathcal{N}} \Big| \sum_{j=1}^m \epsilon_j \abs{\va_j^*\vx}^2\Big| \ge \frac{C\sqrt{md}+t\sqrt{m}}{2} } \le 2\exp\xkh{-c  t}.
\end{equation}
Combining (\ref{eq:operate}) and (\ref{eq:Pup}), we obtain that if $m \gtrsim d$ then
 with probability at least $1-2\exp(-c t)$ it holds
\[
 \frac{1}{\sqrt{m}}\norm{\sum_{j=1}^m \epsilon_j \va_j\va_j^*} \le C\sqrt{d}+t
\]
for all $t\ge 0$.
According to the definition of expectation, we have
\begin{equation*}
\begin{aligned}
&\E~ \frac{1}{\sqrt{m}}\norm{\sum_{j=1}^m \epsilon_j \va_j\va_j^*} = \int_0^\infty \PP\dkh{ \frac{1}{\sqrt{m}}\norm{\sum_{j=1}^m \epsilon_j \va_j\va_j^*} \ge t} dt\\
&= \int_0^{C\sqrt{d}} \PP\dkh{ \frac{1}{\sqrt{m}}\norm{\sum_{j=1}^m \epsilon_j \va_j\va_j^*} \ge t} dt
+\int_0^\infty \PP\dkh{ \frac{1}{\sqrt{m}}\norm{\sum_{j=1}^m \epsilon_j \va_j\va_j^*} \ge C\sqrt{d}+t} dt\\
&\le  C\sqrt{d}+2 \int_0^\infty e^{-ct} dt\lesssim  \sqrt{d}.
\end{aligned}
\end{equation*}
\end{proof}

\subsubsection{Proof of Theorem \ref{le:ORIP}}

We next present a proof of Theorem \ref{le:ORIP}. We would like to mention that
 one can prove Theorem \ref{le:ORIP}
based on RUB condition and the results in \cite{cai2015rop} (see Section 2.2.3 for
details). For completeness, we provide a proof which employs Mendelson's   small-ball
method.

{\noindent\it Proof of Theorem \ref{le:ORIP}}~ According to  Definition \ref{de:rip},
it is sufficient to prove that
\begin{equation}\label{eq:jiemid}
\norm{\A(H)} \gtrsim \sqrt{m} \normf{H} \qquad \mbox{for all} \quad H \in \H_{t\cdot r}(\C)
\end{equation}
holds with high probability. Due to  homogeneity, without loss of generality, we can
assume $\normf{H}=1$. We employ  Mendelson's  small-ball method to prove the
conclusion  (see Lemma \ref{le:lower} and Lemma \ref{le:palzyg}). To see this, we
identify $\H(\C) \in \C^{d\times d}$ with $\R^{d^2}$ and let
$$\mathcal{SH}_{r}^{d\times
d}:=\dkh{H \in \H_r(\C): \normf{H}=1}. $$ For any $\xi \ge 0$ define
\begin{equation*}
\begin{aligned}
Q_\xi&:=\inf_{H\in \mathcal{SH}_{t\cdot r}^{d\times d}} \PP\dkh{\abs{\va_j^*H\va_j}\ge
\xi}\\
 W_m&:= \E\sup_{H\in \mathcal{SH}_{t\cdot r}^{d\times d}} \nj{H,A} \quad
\mbox{where} \quad A:=\frac{1}{\sqrt{m}} \sum_{j=1}^m \epsilon_j \va_j\va_j^*.
\end{aligned}
\end{equation*}
Here, the $\epsilon_1,\ldots,\epsilon_m$ are independent Rademacher random variables.
Then Lemma \ref{le:lower} implies that, with probability at least
$1-\exp(-\gamma^2/2)$, it holds
\begin{equation}\label{eq:lowboundh}
\inf_{H\in \mathcal{SH}_{t\cdot r}^{d\times d}} \|\A(H)\|_2=
\inf_{H\in \mathcal{SH}_{t\cdot r}^{d\times d}} \xkh{\sum_{j=1}^m (\va_j^* H
\va_j)^2}^{1/2} \ge \xi\sqrt{m}Q_{2\xi}-2W_m-\xi \gamma
\end{equation}
for any $\xi>0$ and $\gamma>0$. We take  $\xi=\sqrt{2}/4, \gamma=c\sqrt{m}$ for
a sufficiently small positive constant $c$ in (\ref{eq:lowboundh}) and claim that
\begin{equation}\label{eq:jie1}
Q_{1/\sqrt{2}} \ge \frac{1}{52},\,\, W_m \lesssim \sqrt{tdr}.
\end{equation}
Combing (\ref{eq:lowboundh}), (\ref{eq:jie1}) and  $m\gtrsim tdr$, we arrive at
(\ref{eq:jiemid}).

It remains to prove (\ref{eq:jie1}).  For the term $Q_{1/\sqrt{2}}$, according to the
Payley-Zygmund inequality (Lemma \ref{le:palzyg}), we have
\begin{equation} \label{eq:palzym}
\PP\dkh{\abs{\va_j^*H\va_j}^2\ge \frac{1}{2} \E\abs{\va_j^*H\va_j}^2} \ge \frac{1}{4}\cdot \frac{(\E|\va_j^*H\va_j|^2)^2}{\E|\va_j^*H\va_j|^4}.
\end{equation}
 By spectral decomposition, we can write
$H:=\sum_{j=1}^{t\cdot r}\lambda_j \vv_j\vv_j^*$ where $\lambda_1,\ldots,\lambda_{t\cdot r}
\in \R$ are eigenvalues and $\vv_1,\ldots,\vv_{t\cdot r}\in \Cd$  are the
corresponding orthonormal eigenvectors. For a standard complex Gaussian  random
variable $Z\sim 1/\sqrt{2}\cdot \cN(0,1)+i/\sqrt{2}\cdot \cN(0,1) $, we have
$\E\abs{Z}^{2k}=k!, k\in {\mathbb Z}_+$. By the unitary invariance of complex
Gaussian random vectors, we have
\begin{equation*}
\begin{aligned}
\E\abs{\va_j^*H\va_j}^2&=\E \xkh{\sum_{k=1}^{t\cdot r} \lambda_k \abs{\va_{j,k}}^2}^2=2\sum_{k=1}^{t\cdot r} \lambda_k^2+2\sum_{1\le k<l \le t\cdot r}\lambda_k \lambda_l \\
&=\sum_{k=1}^{t\cdot r} \lambda_k^2 + \xkh{\sum_{k=1}^{t\cdot r} \lambda_k }^2
\ge \normf{H}^2.
\end{aligned}
\end{equation*}
Noting that $\normf{H}^2=\sum_{k=1}^{t\cdot r}\lambda_k^2=1$, we have $|\lambda_k| \le 1$.
Then $|\sum_{k=1}^{t\cdot r}\lambda_k^3|\le \sum_{k=1}^{t\cdot r}\lambda_k^2=1$ and
$\sum_{k=1}^{t\cdot r}\lambda_k^4 \le \sum_{k=1}^{t\cdot r}\lambda_k^2=1$. Thus we have
\begin{equation}\label{eq:EAA}
\begin{aligned}
\E\abs{\va_j^*H\va_j}^4&= \sum_{j\neq k\neq l\neq s}\lambda_j \lambda_k \lambda_l \lambda_s+12\sum_{k\neq l\neq s} \lambda_k^2 \lambda_l \lambda_s+12\sum_{k\neq l}\lambda_k^2 \lambda_l^2+24\sum_{k\neq l}\lambda_k^3 \lambda_l +24\sum_{k}\lambda_k^4\\
&= \sum_{j, k, l, s}\lambda_j \lambda_k \lambda_l \lambda_s+6\sum_{k\neq l\neq s} \lambda_k^2 \lambda_l \lambda_s+9\sum_{k\neq l}\lambda_k^2 \lambda_l^2+20\sum_{k\neq l}\lambda_k^3 \lambda_l +23\sum_{k}\lambda_k^4\\
&=\sum_{j,k,l,s}\lambda_j \lambda_k \lambda_l \lambda_s+6\sum_{k,l,s} \lambda_k^2 \lambda_l \lambda_s+3\sum_{k,l}\lambda_k^2 \lambda_l^2+8\sum_{k,l}\lambda_k^3 \lambda_l +6\sum_{k}\lambda_k^4  \\
&= \Big( \sum_{k=1}^{t\cdot r}\lambda_k\Big)^4+6\normf{H}^2  \Big(\sum_{k=1}^{t\cdot r}\lambda_k\Big)^2+3\normf{H}^4+8\Big(\sum_{k=1}^{t\cdot r}\lambda_k^3\Big) \Big(\sum_{k=1}^{t\cdot r}\lambda_k\Big)+6\sum_{k=1}^{t\cdot r}\lambda_k^4 \\
&\le   \Big( \sum_{k=1}^{t\cdot r}\lambda_k\Big)^4+6\Big(\sum_{k=1}^{t\cdot r}\lambda_k\Big)^2+4\xkh{1+ \Big(\sum_{k=1}^{t\cdot r}\lambda_k\Big)^2}+9\\
&\le  13(\E\abs{\va_j^*H\va_j}^2)^2.
\end{aligned}
\end{equation}
Here, we use the fact that  $2|\sum_{k=1}^{t\cdot r}\lambda_k|\le
1+\Big(\sum_{k=1}^{t\cdot r}\lambda_k\Big)^2$ in the first inequality and
$(\E|\va_j^*H\va_j|^2)^2 =(\sum_{k=1}^{t\cdot r}\lambda_k)^4+2( \sum_{k=1}^{t\cdot
r}\lambda_k)^2+1$ in the last inequality. Putting (\ref{eq:EAA}) into
(\ref{eq:palzym}), we obtain
\[
\PP\dkh{\abs{\va_j^*H\va_j}^2 \ge \frac{1}{2} \normf{H}^2} \ge \frac{1}{52},
\]
which implies
\begin{equation*}
Q_{1/\sqrt{2}} \ge \frac{1}{52}.
\end{equation*}

 We next show $W_m= \E\sup_{H\in \mathcal{SH}_{t\cdot r}^{d\times d}} \nj{H,A}\lesssim \sqrt{tdr}$. For
any $H\in \mathcal{SH}_{t\cdot r}^{d\times d}$, by spectral decomposition, we can
write $H=\sum_{j=1}^{t\cdot r}\lambda_j \vv_j\vv_j^*$. Then
\[
 \nj{H,A}= \sum_{j=1}^{t\cdot r}  \lambda_j \vv_j^*A\vv_j \le \norm{A} \cdot \sum_{j=1}^{t\cdot r}  \abs{\lambda_j} =\norm{A}\norms{H}_{*} \le \sqrt{t\cdot r}\norm{A}.
\]
Here, we use the fact that the nuclear norm $\norms{H}_{*}\le \sqrt{t\cdot r} \normf{H}$ due to
$\rank(H)\le t\cdot r$. It gives
\[
W_m= \E\sup_{H\in \mathcal{SH}_{t\cdot r}^{d\times d}} \nj{H,A} \le \sqrt{t\cdot r}\cdot\E\norm{A}.
\]
Recall that $A=\frac{1}{\sqrt{m}} \sum_{j=1}^m \epsilon_j \va_j\va_j^*$ where
$\epsilon_1,\ldots,\epsilon_m$ are independent Rademacher random variables,
independent from everything else. From Lemma \ref{le:simiA},  we have
\[
\E \norm{\frac{1}{\sqrt{m}} \sum_{j=1}^m \epsilon_j \va_j\va_j^*}\lesssim \sqrt{d},
\]
which implies
\begin{equation*}
W_m \lesssim \sqrt{tdr}.
\end{equation*}
This completes the claim.
\qed
\subsubsection{The connection between LRIP and RUB}

In  \cite{cai2015rop}, Cai and Zhang introduce the definition of  Restricted Uniform
Boundedness (RUB) and prove that the Gaussian rank-one projection satisfies such
condition with high probability in the real case. One can prove Theorem \ref{le:ORIP}
based on RUB condition and the results in \cite{cai2015rop}, as  shown below.


A linear map $\A: \R^{p_1\times p_2}\to \R^m$ has RUB condition of order $r$ if there
exist uniform   constants $C_1$ and $C_2$ such that
\[
C_1\normf{M} \le \frac{1}{m} \normone{\A(M)} \le C_2\normf{M}
\]
holds for all non-zero rank-$r$ matrices $M\in \R^{p_1\times p_2}$.
 Using the notation RUB and the results under `` Rank-One Projection'' model  in \cite{cai2015rop}, we
   can present an alternative proof for Theorem \ref{le:ORIP}.
    To see this,  recognize that the linear map defined in Theorem \ref{le:ORIP}
    is $\A: \H(\C) \to \R^m$ with
\[
[\A(H)]_j=\va_j^* H \va_j, \qquad j=1,\ldots,m.
\]
If we let $\va_j=\va_j^{\Re}+ i \va_j^\Im$,  $H=H^\Re+iH^\Im$ with $H^\Re=(H^\Re)^\T$
and $H^\Im=-(H^\Im)^\T$, then we could rewrite the linear map $\A$ as
\begin{equation}  \label{eq:tildeA}
[\A(H)]_j=\va_j^* H \va_j=\ta_j^\T \tilde{H} \ta_j,  \qquad j=1,\ldots,m,
\end{equation}
where
$\tilde{H}=\left( \begin{array}{ll} H^\Re & -H^\Im \\ H^\Im & H^\Re \end{array}
\right) \in \R^{2d\times 2d}$
 and
  $\ta_j = \left( \begin{array}{l}  \va_j^{\Re}\\
\va_j^\Im \end{array} \right)  \in \R^{2d}$.
Let $\tilde{\A}: \R^{2d\times 2d} \to \R^{\lfloor m/2 \rfloor}$ be an
 operator  which is defined by
\[
[{\tilde{\A} (M) }]_j:=\beta_j^\T M \gamma_j, \qquad j=1,\ldots, \lfloor m/2 \rfloor
\]
with $\beta_j:=\ta_{2j-1}+\ta_{2j}$ and $\gamma_j:=\ta_{2j-1}-\ta_{2j}$.
Since $\va_j^{\Re}, \va_j^\Im \sim N(0,I_d/2)$, it leads to $\beta_j, \gamma_j \sim
N(0,I_{2d})$ with $\beta_j$ and $\gamma_j$ independent because they are Gaussian random vectors and
$\E\nj{\beta_j,\gamma_j}=0$. Thus, the linear map $\tilde{\A}$ is exactly a Gaussian
rank-one projection model as defined in  \cite{cai2015rop}.  Theorem 2.2 in
\cite{cai2015rop} shows that with probability at least $1-\exp(-c m)$ for some
constant $c>0$, $\tilde{\A}$ satisfies RUB condition of order $t\cdot r$ provided $m\gtrsim
tdr$, namely, there exist constants $C_1, C_2>0$ such that
\[
C_1 \normf{M} \le \frac{1}{ \lfloor m/2
\rfloor} \normone{\tilde{\A}(M)} \le C_2 \normf{M}
\]
holds for all rank-$(t\cdot r)$ matrices $M\in \R^{2d \times 2d}$.  It follows from
(\ref{eq:tildeA}) that the connection between $\tilde{\A}$ and $\A$ is
\[
[\A(H)]_{2j-1}-[\A(H)]_{2j} =\tilde{\A}(\tilde{H}), \qquad j=1,\ldots,\lfloor m/2
\rfloor.
\]
This means that, for all $H\in \H_{t\cdot r}(\C)$, directly associated with the Hermitian
$\tilde{H} \in \R^{2d\times 2d}$, we have
\begin{eqnarray*}
\frac{1}{\sqrt{m}} \norm{\A(H)} \ge \frac{1}{m} \normone{\A(H)} &\ge& \frac{1}{m} \sum_{j=1}^{ \lfloor m/2
\rfloor} \Big| [\A(H)]_{2j-1}-[\A(H)]_{2j}  \Big|\\
&= &  \frac{1}{m}  \normone{\tilde{\A}(\tilde{H})}
\ge     \frac{C_1}{2} \normf{\tilde{H}} =\frac{C_1}{\sqrt{2}} \normf{H},
\end{eqnarray*}
which implies the result in Theorem \ref{le:ORIP}.

\section{Proofs of  Theorem \ref{result1} and Theorem \ref{th:lowbound} }

Motivated by the observation
\begin{equation} \label{eq:mealift}
b_j=\abs{\nj{\va_j,\vx_0}}^2+\eta_j=\va_j^* X_0 \va_j +\eta_j, \quad j=1,\ldots,m,
\end{equation}
where $X_0:=\vx_0 \vx_0^*$, we can employ  Corollary \ref{th:liftresult}
 to prove Theorem \ref{result1}.

{\noindent\it Proof of Theorem \ref{result1}.}~ Let $X_0:=\vx_0 \vx_0^*$. Since $\x$
is the global solution to (\ref{eq:mod1}), $\X:=\x\x^*$ is the global solution to
(\ref{mod:determ}) with $r=1$. From Corollary \ref{th:liftresult}, we obtain that
with probability  at least $1- \exp(-cm)$, it holds
\begin{equation} \label{eq:matrxuv}
\normf{\x\x^*-\vx_0\vx_0^*} \lesssim \frac{\norm{\eta}}{\sqrt{m}}
\end{equation}
provided $m\gtrsim d$. We claim that for any $\vu,\vv\in \Cd$, we have
\begin{equation} \label{cla:uv}
 \min_{\theta\in [0,2\pi)} \norm{\vu-e^{i\theta}\vv} \le \frac{2\normf{\vu\vu^*-\vv\vv^*}}{\norm{\vv}}.
\end{equation}
Indeed, choosing $\theta:=\mbox{Phase}(\vv^* \vu)$ and setting $\vbar:=e^{i\theta}\vv $, then  $\nj{\vu,\vbar}\ge 0$. Let $\vh:=\vu-\vbar$. Then we have
\begin{eqnarray*}
\normf{\vu\vu^*-\vv\vv^*}^2&=& \normf{\vu\vu^*-\vbar\vbar^*}^2 = \normf{\vh\vh^*+\vh\vbar^*+\vbar\vh^*}^2 \\
&=& \norm{\vh}^4 +4 \norm{\vh}^2 \nj{\vh,\vbar}+2\nj{\vh,\vbar}^2+2\norm{\vh}^2\norm{\vbar}^2\\
&\ge & (4-2\sqrt{2}) \norm{\vh}^2 \nj{\vh,\vbar}+ 2\norm{\vh}^2\norm{\vbar}^2\\
&=& (4-2\sqrt{2}) \norm{\vh}^2  \nj{\vu,\vbar}+ (2\sqrt{2}-2)\norm{\vh}^2\norm{\vbar}^2\\
&\ge & \frac{1}{2}\norm{\vh}^2\norm{\vbar}^2,
\end{eqnarray*}
where the last line follows from the fact $\nj{\vu,\vbar}\ge 0$. Thus we obtain
(\ref{cla:uv}). Combining (\ref{eq:matrxuv}) and (\ref{cla:uv}), we arrive at
\begin{equation}\label{eq:relationHh}
 \min_{\theta\in [0,2\pi)} \norm{\x-e^{i\theta}\vx_0} \le \frac{2}{\norm{\vx_0}}\normf{\x\x^*-\vx_0\vx_0^*}\lesssim
 \frac{ \norm{\eta}}{\norm{\vx_0}\sqrt{m}}.
\end{equation}
We next show
\[
 \min_{\theta\in [0,2\pi)} \norm{\x-e^{i\theta}\vx_0} \lesssim
  \|\vx_0\|_2+\frac{\sqrt{\|\eta\|_2}}{m^{1/4}}.
\]
Indeed, according to (\ref{eq:matrxuv}), we have
\[
\|\x\|_2^2 =\normf{\x\x^*}\leq \normf{\vx_0\vx_0^*}+\normf{\x\x^*-\vx_0\vx_0^*}\lesssim \|\vx_0\|_2^2+\frac{\|\eta\|_2}{\sqrt{m}},
\]
which implies
\[
\|\x\|_2\lesssim \|\vx_0\|_2+\frac{\sqrt{\|\eta\|_2}}{m^{1/4}}.
\]
Hence, we have
\begin{equation}\label{eq:14up}
 \min_{\theta\in [0,2\pi)} \norm{\x-e^{i\theta}\vx_0} \leq \norm{\x}+\norm{\vx_0}\lesssim \|\vx_0\|_2+\frac{\sqrt{\|\eta\|_2}}{m^{1/4}}.
\end{equation}
Combining (\ref{eq:relationHh}) and (\ref{eq:14up}), we obtain
\[
\min_{\theta\in [0,2\pi)} \norm{\x-e^{i\theta}\vx_0}\lesssim \min
\left\{ \|\vx_0\|_2+\frac{\sqrt{\|\eta\|_2}}{m^{1/4}}, \frac{ \norm{\eta}}{\norm{\vx_0}\sqrt{m}}\right\}.
\]
  A simple calculation shows that
\[
\min\left\{\|\vx_0\|_2+\frac{\sqrt{\|\eta\|_2}}{{m}^{1/4}},
\frac{\|\eta\|_2}{\norm{\vx_0}\cdot \sqrt{m}}\right\}=\frac{\|\eta\|_2}{\norm{\vx_0}\cdot
\sqrt{m}}
\]
 holds provided $\|\vx_0\|_2\geq \frac{\sqrt{5}-1}{2}\cdot \frac{\sqrt{\|\eta\|_2}}{m^{1/4}}$.  For the case where
 $\|\vx_0\|_2< \frac{\sqrt{5}-1}{2}\cdot \frac{\sqrt{\|\eta\|_2}}{m^{1/4}}$, we have
\[
\min\left\{\|\vx_0\|_2+\frac{\sqrt{\|\eta\|_2}}{{m}^{1/4}},
\frac{\|\eta\|_2}{\norm{\vx_0}\cdot
\sqrt{m}}\right\}=\|\vx_0\|_2+\frac{\sqrt{\|\eta\|_2}}{{m}^{1/4}}\leq
\frac{\sqrt{5}+1}{2}\cdot \frac{\sqrt{\|\eta\|_2}}{{m}^{1/4}}.
\]
Hence, we obtain
\[
\min_{\theta\in [0,2\pi)} \norm{\x-e^{i\theta}\vx_0}\lesssim \min
\left\{ \|\vx_0\|_2+\frac{\sqrt{\|\eta\|_2}}{m^{1/4}}, \frac{ \norm{\eta}}{\norm{\vx_0}\sqrt{m}}\right\}
\lesssim \min
\left\{ \frac{\sqrt{\|\eta\|_2}}{m^{1/4}}, \frac{ \norm{\eta}}{\norm{\vx_0}\sqrt{m}}\right\}.
\]
  \qed

\subsection{Proof of Theorem \ref{th:lowbound}}

We first introduce some lemmas which play a key role in our proof. The following
lemma is a nonuniform result about the upper bound of the fourth power of complex
Gaussian variables.
\begin{lemma}\cite[Lemma 21]{turstregion} \label{le:sunju}
Let $\va_j \in \C^d, j=1,\ldots,m,$ be i.i.d complex Gaussian random vectors. Suppose
that  $\vv\in \Cd$ is a fixed vector.
 For any $\delta \in (0,1)$ the following holds with
probability at least $1-c_a \delta^{-2} m^{-1}-c_b \exp(-c_c \delta^2 m /\log m)$
\[
\left\| \frac{1}{m} \sum_{j=1}^m \abs{\va_j^* \vv}^2  \va_j \va_j^*- \big( \vv\vv^* +\norm{\vv}^2 ~ I \big)    \right\|          \le  \delta \norms{\vv}^2
\]
provided $m\ge C(\delta) d \log d$. Here  $C(\delta)$ is a constant depending on
$\delta$ and $c_a$, $c_b$ and $c_c$ are positive absolute constants.
\end{lemma}

\begin{lemma}\cite[Lemma 22]{turstregion} \label{le:sunju22}
Let $\va_j \in \C^d, j=1,\ldots,m,$ be i.i.d complex Gaussian random vectors. For any
$\delta \in (0,1)$  the following holds  with probability at least $1-c'  m^{-d}-c''
\exp(-c(\delta) m )$
\[
\frac{1}{m} \sum_{j=1}^m \abs{\va_j^* \vz}^2 \abs{\va_j^* \bm{w}}^2  \ge (1-\delta) \xkh{\norms{\bm{w}}^2\norms{\vz}^2+\abs{\bm{w}^*\vz}^2} \quad \mbox{for all} \quad \vz,\bm{w} \in \C^d
\]
provided $m\ge C(\delta) d \log d$.
 Here  $C(\delta)$ and $c(\delta)$ are constants depending on $\delta$ and
$c'$, $c''$ are positive absolute constants.
\end{lemma}

\begin{lemma} \label{le:etalowbound}
Suppose that $\va_j \in \C^d, j=1,\ldots,m $, are i.i.d. complex Gaussian random vectors and
$\eta_j\in {\mathbb R}, j=1,\ldots,m$. For any fixed $\delta \in (0,1)$,
there exists a constant  $ \rho>0$ depending only on $\delta$  such that the following holds
 with probability at least $1-2\exp(-c (\delta) d)$ :
 \[
   \norm{\sum_{j=1}^m \eta_j (\va_j\va_j^*-I)} \le \rho\cdot\delta\cdot (\sqrt{d}\norm{\eta}+d \norms{\eta}_{\infty}).
 \]
Here, $c(\delta)>0$ is a constant depending only on $\delta$.
\end{lemma}

\begin{proof}
We assume that $\mathcal{N}$
 is a $1/4$-net  of the complex unit sphere ${\mathbb S}^{d-1}\subset {\mathbb C}^d$. It then follows from \cite[Lemma 4.4.3]{Vershynin2018} that
\begin{equation*}
 \norm{\sum_{j=1}^m \eta_j (\va_j\va_j^*-I)} \le 2 \max_{\vx \in \mathcal{N}} \Big| \sum_{j=1}^m \eta_j (| \va_j^*\vx|^2-1)\Big|,
\end{equation*}
where the cardinality $|\mathcal{N}|\le 9^{2d}$.
For any fixed $\vx \in {\mathbb S}^{d-1}$,  the terms $ |\va_j^*\vx|^2-1, j=1,\ldots,m$ are
independent centered sub-exponential random variables with the sub-exponential norm
being a constant.  Using the Bernstein's inequality \cite[Theorem
2.8.1]{Vershynin2018}, we obtain that
\begin{equation*}\label{eq:bipp}
\PP\dkh{ \Big|\sum_{j=1}^m  \eta_j (| \va_j^*\vx|^2-1) \Big| \ge t} \le  2\exp\xkh{-c' \min\Big(\frac{t^2}{\norms{\eta}^2}, \frac t {\norms{\eta}_{\infty}}\Big)}
\end{equation*}
for some positive constant $c'$.
We assume that  $\rho>1$ is a constant which will be specified later.
Taking $t:=\rho(\sqrt{d}\norm{\eta}+d
\norms{\eta}_{\infty})\delta$, we obtain that
\[
 \Big|\sum_{j=1}^m  \eta_j (| \va_j^*\vx|^2-1) \Big|  \le \rho (\sqrt{d}\norm{\eta}+d \norms{\eta}_{\infty})\delta
 \]
holds with probability at least $1-2\exp(-c' \cdot \rho\cdot \delta^2 d)$.  Choosing the constant $\rho$ such that $c' \cdot \rho\cdot \delta^2 \ge 2\ln 9$ and taking the union bound,  we can obtain that, with  probability at least $1-2\exp\xkh{-c (\delta) d}$, it holds that
\[
   \norm{\sum_{j=1}^m \eta_j (\va_j\va_j^*-I)} \le \rho(\sqrt{d}\norm{\eta}+d \norms{\eta}_{\infty})\delta
 \]
 where $c(\delta):=c'\rho\delta^2-2\ln 9>0$ is a constant  depending only on $\delta$.
\end{proof}

The following lemma states that if $\x$ is the solution to (\ref{eq:mod1}),  then the
upper-tail of $\frac 1m \sum_{j=1}^m |\va_j^* \x|^4$ is well-behaved, although it involves the fourth power of the Gaussian variables. We present its
proof in Appendix A.
\begin{lemma} \label{le:upow4}
Suppose that $\va_j \in\Cd, j=1,\ldots,m $, are i.i.d. complex Gaussian random
vectors. Assume that the noise vector $\eta$ satisfies $\norm{\eta}  \asymp \sqrt{m}$,
$\norms{\eta}_{\infty}\lesssim \log m$ and $~\sum_{j=1}^m \eta_j \lesssim m$. Suppose
that $\vx_0\in \Cd$  is a fixed vector satisfying  $\norm{\vx_0} \ge
2C\norm{\eta}/\sqrt{m}$.
 For any $\gamma>0$,   if $m\ge c'(\gamma) d \log m$ then with probability at least
$1-c''(\gamma) m^{-1}-c''' \exp(-c''''(\gamma) d)$  the following holds: any solution
$\x \in \Cd$ to (\ref{eq:mod1}) with $b_j=\abs{\nj{\va_j,\vx_0}}^2+\eta_j, \;
j=1,\ldots,m$, satisfies
\[
\frac 1m \sum_{j=1}^m \abs{\va_j^* \x}^4 \le (2+ \gamma)\norms{\x}^4.
\]
Here $c'(\gamma),c''(\gamma), c''''(\gamma)$ are constants depending on $\gamma$,  and $c'''$ is an absolute constant,  and  $C$ is the constant in Theorem \ref{result1}.
\end{lemma}


We next present the proof of Theorem \ref{th:lowbound}.
\begin{proof}[Proof of Theorem \ref{th:lowbound}]~
Without loss of generality, we assume $\norm{\vx_0}=1$ (the general case can be obtained via a simple rescaling).
Let
\[
f(\vx)=\sum_{j=1}^m \xkh{\abs{\nj{\va_j,\vx}}^2-b_j}^2.
\]
Then the Wirtinger gradient (see, eg, \cite{WF}) of $f$ is
\[
\nabla f(\vx)=2\sum_{j=1}^m \xkh{\abs{\nj{\va_j,\vx}}^2-b_j} \va_j\va_j^*\vx.
\]
Since $\x $ is the solution to  (\ref{eq:mod1}), we have
\begin{equation*}
\nabla f(\x)=2\sum_{j=1}^m \xkh{\abs{\nj{\va_j,\x}}^2-b_j} \va_j\va_j^*\x=0.
\end{equation*}
Noting that $b_j=\abs{\nj{\va_j,\vx_0}}^2+\eta_j, \;  j=1,\ldots,m$, then $\nj{\nabla f(\x),\x}=0$ gives
\begin{equation} \label{eq:1lowerkey}
\frac1m \sum_{j=1}^m \abs{\va_j^* \x}^4 -\frac1m \sum_{j=1}^m \abs{\va_j^* \x}^2\abs{\va_j^* \vx_0}^2= \frac1m \sum_{j=1}^m \eta_j \abs{\va_j^* \x}^2.
\end{equation}
We claim that, when $m\gtrsim d \log m$, with probability at least $1-c' m^{-1}-c''
\exp(-c''' d)$ it holds:  any solution $\x$ to  (\ref{eq:mod1}) obeys
\begin{equation} \label{eq:up5}
\abs{\frac1m \sum_{j=1}^m \abs{\va_j^* \x}^4 -\frac1m \sum_{j=1}^m \abs{\va_j^* \x}^2\abs{\va_j^* \vx_0}^2} \le 3\norms{\x}^4-\norms{\x}^2 \le 3\norms{\x}^2 \normf{H}
\end{equation}
and
\begin{equation} \label{eq:infbound5}
 \frac1m \Big| \sum_{j=1}^m \eta_j \abs{\va_j^*\x}^2\Big| \gtrsim \norm{\eta} \norm{\x}^2/\sqrt{m},
 \end{equation}
where $H=\x\x^*-\vx_0\vx_0^*$ and $ c', c'', c'''$ are universal positive constants.
Combining (\ref{eq:1lowerkey}), (\ref{eq:up5}) and (\ref{eq:infbound5}), we obtain
that
\begin{equation}\label{eq:Hlow}
\normf{H}=\normf{\x\x^*-\vx_0\vx_0^*} \gtrsim \frac{\norm{\eta}}{\sqrt{m}}
\end{equation}
holds with probability at least $1-c' m^{-1}-c'' \exp(-c''' d)$ provided $m\gtrsim d
\log m$. We next use (\ref{eq:Hlow}) to derive the conclusion.   For any $\theta \in
[0,2\pi)$, we have
\begin{equation}\label{eq:Hlow1}
\begin{aligned}
\normf{\x\x^*-\vx_0\vx_0^*} &=  \normf{\x\x^*-\x\vx_0^*e^{-i\theta}+\x\vx_0^*e^{-i\theta}-\vx_0\vx_0^*}\\
&\le  \norm{\x} \norm{\x-\vx_0e^{i\theta}}+  \norm{\vx_0} \norm{\x-\vx_0e^{i\theta}}.
\end{aligned}
\end{equation}
If $\|\vx_0\|_2< \frac{\sqrt{\|\eta\|_2}}{{m}^{1/4}}$, then Theorem \ref{result1} gives
\begin{equation} \label{eq:boundxhat}
\norm{\x}\lesssim \norm{\vx_0}+\frac{\sqrt{\|\eta\|_2}}{{m}^{1/4}} <\frac{2\sqrt{\|\eta\|_2}}{{m}^{1/4}}.
\end{equation}
Combining (\ref{eq:Hlow}), (\ref{eq:Hlow1}) and (\ref{eq:boundxhat}), we obtain that
\[
 \min_{\theta\in [0,2\pi)} \norm{\x-e^{i\theta}\vx_0} \gtrsim \frac{{m}^{1/4}}{\sqrt{\norm{\eta}}}\cdot \normf{\x\x^*-\vx_0\vx_0^*}\gtrsim   \frac{\sqrt{\|\eta\|_2}}{{m}^{1/4}}.
\]
On the other hand, if $\|\vx_0\|_2\ge \frac{\sqrt{\|\eta\|_2}}{{m}^{1/4}}$, then Theorem \ref{result1} gives
\[
\norm{\x}\lesssim \norm{\vx_0}+  \frac{\|\eta\|_2}{\norm{\vx_0}\cdot \sqrt{m}} \le  2\norm{\vx_0}.
\]
According to  (\ref{eq:Hlow}) and (\ref{eq:Hlow1}), we obtain
\[
 \min_{\theta\in [0,2\pi)} \norm{\x-e^{i\theta}\vx_0} \gtrsim \frac{1}{\norm{\vx_0}}\normf{\x\x^*-\vx_0\vx_0^*}\gtrsim  \frac{ \norm{\eta}}{\norm{\vx_0}\sqrt{m}}.
\]
We arrive at the conclusion.

 It remains to prove (\ref{eq:up5}) and
(\ref{eq:infbound5}). We first consider (\ref{eq:up5}). According to Lemma
\ref{le:upow4},  if $m\gtrsim d \log m$ then the following holds
\begin{equation}\label{eq:di1}
\frac 1m \sum_{j=1}^m \abs{\va_j^* \x}^4 \le 3\norms{\x}^4
\end{equation}
with probability at least $1-c_1 m^{-1}-c_2 \exp(-c_3 d)$ for some positive constants
$c_1, c_2$ and $c_3$. On the other hand, Lemma  \ref{le:sunju} implies  that if
$m\gtrsim d \log d$ then with probability at least $1-c_1 m^{-1}-c_2 \exp(-c_3 m/\log
m)$ it holds
\begin{equation}\label{eq:di2}
\frac1m \sum_{j=1}^m \abs{\va_j^* \x}^2\abs{\va_j^* \vx_0}^2 \ge \norms{\x}^2.
\end{equation}
Combining (\ref{eq:di1}) and (\ref{eq:di2}), we obtain that,  with probability at
least $1-2c_1 m^{-1}-2c_2 \exp(-c_3 d)$,
\begin{equation}\label{eq:1abs}
\frac1m \sum_{j=1}^m \abs{\va_j^* \x}^4 -\frac1m \sum_{j=1}^m \abs{\va_j^* \x}^2\abs{\va_j^* \vx_0}^2 \le 3\norms{\x}^4-\norms{\x}^2 \le 3\norms{\x}^2 \normf{H}
\end{equation}
holds provided    $m\gtrsim d \log m$. Similarly, we can use  Lemma \ref{le:sunju}
and Lemma \ref{le:sunju22} to obtain that
\begin{equation}\label{eq:2abs}
 \frac1m \sum_{j=1}^m \abs{\va_j^* \x}^2\abs{\va_j^* \vx_0}^2 -  \frac1m \sum_{j=1}^m \abs{\va_j^* \x}^4 \le 3\norms{\x}^2-\norms{\x}^4 \le 3\norms{\x}^2 \normf{H}
\end{equation}
holds with probability at least $1-2c_1 m^{-1}-2c_2 \exp(-c_3 d)$ provided $m\gtrsim
d \log m$. Combining (\ref{eq:1abs}) and (\ref{eq:2abs}), we arrive at
\eqref{eq:up5}.

 We next consider (\ref{eq:infbound5}).
To prove (\ref{eq:infbound5}), it is enough to show that the following holds with
high probability:
\begin{equation} \label{eq:newinfbound5}
  \Big| \sum_{j=1}^m \eta_j \abs{\va_j^*\vx}^2\Big| \gtrsim \sqrt{m}\norm{\eta} \quad \mbox{for all} \quad \vx \in \Cd \text{ with } \|\vx\|=1.
 \end{equation}
  A simple observation is that  for any fixed $\vx\in \C^d$ with  $\|\vx\|=1$
  the terms $|\va_j^*\vx|^2$ are independent
sub-exponential random variables with the sub-exponential norm $C_0$ where $C_0$ is a
constant. According to Bernstein's inequality, we have
\begin{equation}\label{eq:bern12}
\PP\dkh{\Big| \sum_{j=1}^m \eta_j \big(\abs{\va_j^*\vx}^2-1\big)\Big|\ge t}
\le 2\exp\xkh{-c_4 \min \Big(\frac{t^2}{C_0^2 \norms{\eta}^2 },
\frac{t}{C_0 \norms{\eta}_{\infty}  }\Big)}.
\end{equation}
Assume that $c_5>1$ is a constant which will be specified later.
Taking $t:= \sqrt{c_5 d }\norm{\eta} + c_5 d \norms{\eta}_{\infty}$, we have
\[
\frac{t^2}{C_0^2 \norms{\eta}^2 } \ge  \frac{c_5 }{C_0^2} d \qquad \mbox{and} \qquad
 \frac{t}{C_0 \norms{\eta}_{\infty}  } \ge  \frac{c_5 }{C_0} d.
\]
Then (\ref{eq:bern12}) implies that, with probability at least
\[
1-2\exp\xkh{-c_4 \min \Big(\frac{t^2}{C_0^2 \norms{\eta}^2 },
\frac{t}{C_0 \norms{\eta}_{\infty}  }\Big)} \ge 1-2\exp\xkh{-c_4c_5 d \cdot \min \Big(\frac{1}{C_0^2}, \frac{1}{C_0 }\Big)},
\]
 it
holds that
\begin{eqnarray*}
\Big| \sum_{j=1}^m \eta_j \abs{\va_j^*\vx}^2\Big| &\ge &
 { \big|\sum_{j=1}^m \eta_j\big|- \sqrt{c_5 d}\norm{\eta} - c_5 d \norms{\eta}_{\infty} }.
\end{eqnarray*}
Hence, for any fixed $c_5>0$ there exists a constant $C'>0$ such that if $m\ge C'd\log m$ then
\begin{equation} \label{eq:eta5}
 \Big| \sum_{j=1}^m \eta_j \abs{\va_j^*\vx}^2\Big|   \gtrsim \sqrt{m}\norm{\eta}
\end{equation}
holds with probability at least $1-2\exp(-c d)$, where $ c:=c_4 c_5\min\xkh{1/C_0^2,1/C_0}$.
Here, we use  $\big| \sum_{j=1}^m \eta_j\big| \gtrsim \sqrt{m}\norm{\eta} \gtrsim \sqrt{C'} d
\norms{\eta}_{\infty} $, $\norm{\eta} \gtrsim \sqrt{m}$ and $\norms{\eta}_{\infty} \lesssim \log m$.

 Next, we give a uniform bound
for (\ref{eq:eta5}).  We assume that $\mathcal{N}$ is an $\epsilon$-net of the unit
complex sphere in $\Cd$  and hence the covering number $\# \mathcal{N}\le
(1+\frac{2}{\epsilon})^{2d}$. For any $\vx' \in \Cd$ with $\norm{\vx'}=1$, there exists
a $\vx \in \mathcal{N}$ such that $\norm{\vx'-\vx}\le \epsilon$.  Taking $\delta=1/2$ in Lemma
\ref{le:etalowbound}, we can obtain that if $m\gtrsim d \log m$ then the following holds
with probability at least $1-2\exp(-c_6 d)$
\begin{equation}\label{eq:epsnet5}
\begin{aligned}
\Big| \sum_{j=1}^m \eta_j \abs{\va_j^*\vx'}^2\Big|-\Big| \sum_{j=1}^m \eta_j \abs{\va_j^*\vx}^2\Big|
 &\le  \Big| \sum_{j=1}^m \eta_j \va_j^*( \vx'\vx'^*-\vx\vx^*)\va_j \Big| \\
 &\lesssim (\sqrt{d}\norm{\eta} + d\norms{\eta}_{\infty})\normf{\vx'\vx'^*-\vx\vx^*}\\
&\lesssim (\sqrt{m}\norm{\eta} + d\norms{\eta}_{\infty})\normf{\vx'\vx'^*-\vx\vx^*}\\
&\lesssim (\sqrt{m}\norm{\eta} + d\norms{\eta}_{\infty})\epsilon  \\
&\lesssim \epsilon \sqrt{m}\norm{\eta},
\end{aligned}
\end{equation}
where we take $\epsilon$ to be some positive constant  in the third inequality and use the fact that $ \sqrt{m}\norm{\eta} \gtrsim d
\norms{\eta}_{\infty} $  in the last inequality. Here, $c_6$ is a positive constant.  We can choose the constant $c_5$ such that
\[
c=c_4 c_5\min\xkh{1/C_0^2,1/C_0}> 2\log(1+2/\epsilon).
\]
Combining (\ref{eq:eta5}) and
(\ref{eq:epsnet5}), we obtain that (\ref{eq:newinfbound5}) holds  with probability at
least
\[
1-2\exp(-c m)\cdot (1+\frac{2}{\epsilon})^{2d}-2\exp(-c_6 d) \ge 1-4\exp(-c_7 d)
\]
for some positive constant $c_7$, provided $m\gtrsim d\log m$.
\end{proof}

\section{Proof of Theorem \ref{result sparse}}
In this section, we present the proof of Theorem \ref{result sparse}. Before
proceeding, we introduce several notations and technical lemmas which are useful in
the proof. For convenience, we set
\[
S_{d,s}:=\dkh{\vx\in \Cd: \norm{\vx} \le 1, \norms{\vx}_0 \le s},
\]
and
\[
K_{d,s}:=\dkh{\vx\in \Cd: \norm{\vx}\le 1, \normone{\vx} \le \sqrt{s}}.
\]
The relationship of the two sets are characterized by the following lemma:
\begin{lemma} \cite[Lemma 3.1]{Plan2013} \label{le:relaks}
It holds that $\mathrm{conv}(S_{d,s}) \subset K_{d,s} \subset
2\mathrm{conv}(S_{d,s})$,  where $\mathrm{conv}(S_{d,s})$ denotes the convex hull of
$S_{d,s}$.
\end{lemma}
The following lemma presents an upper bound for the spectral norm of $\A^*(\epsilon)$
which is defined in (\ref{eq:dualo}).
\begin{lemma} \label{le:sparsedual}
Suppose that $\va_j\in \C^s,j=1,\ldots,m $, are i.i.d.  complex Gaussian random
vectors and $\epsilon_1,\ldots,\epsilon_m$ are independent Rademacher random
variables. If $m \gtrsim  s$, then with probability at least $1-2\exp(-c m)$ we have
$ \norm{\A^*(\epsilon)} \lesssim m$, that is,
\[
\norm{\sum_{j=1}^m \epsilon_j \va_j\va_j^*} \lesssim  m,
\]
where $\epsilon:=(\epsilon_1,\ldots,\epsilon_m)^\T$.
\end{lemma}

\begin{proof}
Assume that $\mathcal{N}$  is a $1/4$-net  of the complex unit sphere ${\mathbb
S}^{s-1}\subset {\mathbb C}^s$. Then we have
\begin{equation*}
 \norm{\sum_{j=1}^m \epsilon_j \va_j\va_j^*} \le 2 \max_{\vx \in \mathcal{N}} \Big| \sum_{j=1}^m \epsilon_j \abs{\va_j^*\vx}^2\Big|.
\end{equation*}
 For any fixed $\vx \in \mathcal{N}$,  the terms $\epsilon_j |\va_j^*\vx|^2, j=1,\ldots,m$ are
independent centered sub-exponential random variables with maximum sub-exponential
norm $C_1$,  where $C_1$ is a constant.  We use Bernstein's inequality \cite[Theorem
2.8.1]{Vershynin2018} to obtain
\begin{eqnarray*}
\PP\dkh{ \Big| \sum_{j=1}^m \epsilon_j \abs{\va_j^*\vx}^2\Big| \ge Cm }   \le 2\exp\xkh{-c' \min \Big( \frac{C^2 m^2}{C_1^2m}, \frac{C m}{ C_1}\Big)}\le 2\exp\xkh{-c m}
\end{eqnarray*}
for some positive constants $C,c,c'$. Recognize that $|\mathcal{N}|\le 9^{2s}$.
Taking the union bound over $\mathcal{N}$,  we can obtain that for $m \gtrsim s$ with
probability at least $1-2\exp\xkh{-c m}$ it holds
\[
 \norm{\sum_{j=1}^m \epsilon_j \va_j\va_j^*} \lesssim m,
\]
which completes the proof.
\end{proof}

Following the spirit of LRIP condition, we need  to present a lower bound
$\sum_{j=1}^m (\va_j^* H \va_j)^2$ for simultaneously sparse and low-rank matrix $H$
under the optimal sampling complexity, as demonstrated in the theorem below.


\begin{lemma} \label{le:cruciallem}
Suppose that $\va_j\in \C^d,j=1,\ldots,m, $ are i.i.d. complex Gaussian random
vectors and $\epsilon_1,\ldots,\epsilon_m$ are independent Rademacher random
variables. Assume that $\eta\in \R^m$ is a vector. If $m\gtrsim s \log (ed/s)$, then
the followings hold  with probability at least $1-4\exp(-c m)$:
\begin{itemize}
\item[(i)]  There exists a sufficiently small constant $c_0>0$ such that
\begin{equation*}
 \big| \sum_{j=1}^m \epsilon_j \va_j^*\vu \vv^* \va_j \big |\le c_0m \quad
\text{ for all }
    \vu,\vv \in K_{d,s}.
 \end{equation*}
%
\item[(ii)]
\begin{equation*}
\inf_{H\in  M_{d,s}}\norm{\A(H)} \gtrsim  \sqrt{m},
\end{equation*}
where
\[
\begin{array}{l}
M_{d,s}:=\left\{\frac{\vh\vh^*+\vh\vx^*+\vx\vh^*}{\normf{\vh\vh^*+\vh\vx^*+\vx\vh^*}}\in \H_2(\C):  \vh/\norm{\vh} \in K_{d,s},\;  \vx \in K_{d,s} \quad \mbox{and}  \right.\vspace{1em} \\
\qquad \qquad \quad   \norm{\vh}\le 2\normf{\vh\vh^*+\vh\vx^*+\vx\vh^*} \Big\}.
\end{array}
\]
\end{itemize}
\end{lemma}
\begin{proof}
We first prove (i).  According to Lemma \ref{le:relaks} we have
\[
\sup_{\vu,\vv \in K_{d,s}}\Big|  \sum_{j=1}^m \epsilon_j \va_j^* \vu \vv^* \va_j \Big|\lesssim \sup_{\vu,\vv \in S_{d,s}} \Big| \sum_{j=1}^m \epsilon_j \va_j^* \vu \vv^* \va_j \Big |.
\]
It suffices  to present an upper bound for  $\sup_{\vu,\vv \in S_{d,s}} \big|
\sum_{j=1}^m \epsilon_j \va_j^* \vu \vv^* \va_j\big |$. For any fixed $\vu_0,\vv_0\in
S_{d,s}$, the terms $ \epsilon_j \va_j^* \vu_0 \vv_0^* \va_j$ are independent
centered sub-exponential random random variables with maximum sub-exponential norm
$C$, where $C$ is a constant. The Bernstein's inequality gives
\begin{equation}\label{eq:etaaxsparse}
\PP\dkh{ \Big| \sum_{j=1}^m \epsilon_j \va_j^* \vu_0 \vv_0^* \va_j \Big |\ge c_1 m} \le 2\exp\xkh{-c' \min \Big( \frac{c_1^2 m^2}{C^2m}, \frac{c_1 m}{ C}\Big)} \le 2\exp(-c m)
\end{equation}
for some sufficiently small constants $c, c_1,c'>0$. Suppose that $\mathcal{N}$ is an
$\epsilon$-net of $S_{d,s}\times S_{d,s}$. Hence, for any $\vu,\vv \in S_{d,s}$,
there exist $\vu_0,\vv_0 \in \mathcal{N}$ satisfying $\norm{\vu-\vu_0}\le \epsilon$
and $\norm{\vv-\vv_0}\le \epsilon$. Note that  the matrix $\vu \vv^*-\vu_0 \vv_0^*$
has at most $2s$ nonzero columns and $2s$ nonzero rows because of $
\vu,\vv,\vu_0,\vv_0 \in S_{d,s}$. Using Lemma \ref{le:sparsedual}, we  obtain that if
$m\gtrsim 2s $ then, with probability at least $1-2\exp(-c m)$, it holds
\begin{equation}\label{eq:epsnetspars}
\begin{aligned}
&\abs{\Big|  \sum_{j=1}^m \epsilon_j \va_j^* \vu \vv^* \va_j\Big|-\Big| \sum_{j=1}^m \epsilon_j \va_j^* \vu_0 \vv_0^* \va_j\Big| }
\le  \Big| \sum_{j=1}^m \epsilon_j \va_j^*(\vu \vv^*-\vu_0 \vv_0^*)\va_j \Big| \\
& \le \left|\nj{\vu \vv^*-\vu_0 \vv_0^* , \sum_{j=1}^m \epsilon_j \va_j\va_j^*}\right|
\leq \sqrt{2}\norm{\sum_{j=1}^m \epsilon_j \va_j\va_j^*} \normf{\vu \vv^*-\vu_0 \vv_0^*} \\
&\lesssim  m \normf{\vu \vv^*-\vu_0 \vv_0^*}
\lesssim m \epsilon,
\end{aligned}
\end{equation}
 where we use  $ \nj{A,B} \le \sqrt{r} \norm{B}\normf{A}$ for
  any Hermitian matrices $A,B$ with $\rank(A)\le r$.
 Note that the covering number $\# \mathcal{N} \le \exp(Cs\log(ed/s)/\epsilon^2)$.  Choosing a
 sufficiently small constant $\epsilon>0$ and
taking the union bound over $\mathcal{N}$,  we obtain that if $m\gtrsim s \log
(ed/s)$  then with probability at least $1-4\exp(-c m)$, it holds
\[
\sup_{\vu,\vv \in S_{d,s}} \Big|  \sum_{j=1}^m \epsilon_j \va_j^* \vu \vv^* \va_j\Big| \le c_0 m
\]
 for some sufficiently small positive constant $c_0$. Here, we use
(\ref{eq:etaaxsparse}) and (\ref{eq:epsnetspars}). This completes the proof of (i).

We next turn to prove (ii). Let
\begin{equation*}
\begin{aligned}
Q_\xi&:=\inf_{H\in M_{d,s}} \PP\dkh{\abs{\va_j^*H\va_j}\ge \xi}\\
W_m&:= \E\sup_{H\in M_{d,s}} \nj{H,A} \quad \mbox{where} \quad A:=\frac{1}{\sqrt{m}} \sum_{j=1}^m \epsilon_j \va_j\va_j^*.
\end{aligned}
\end{equation*}
Here, the $\epsilon_1,\ldots,\epsilon_m$ are independent Rademacher random variables.
Then Lemma \ref{le:lower} shows that, with probability at least $1-\exp(-t^2/2)$, it holds
\begin{equation}\label{eq:lowboundhsparse}
\inf_{H\in  M_{d,s}} \xkh{\sum_{j=1}^m (\va_j^* H \va_j)^2}^{1/2} \ge \xi\sqrt{m}Q_{2\xi}-2W_m-\xi t
\end{equation}
for any $\xi>0$ and $t>0$. Taking $\xi=\sqrt{2}/4$, we can employ the method in  the
proof of  Theorem   \ref{le:ORIP} to obtain
\begin{equation}\label{eq:qxisparse}
Q_{1/\sqrt{2}} \ge \frac{1}{52}.
\end{equation}
We next proceed to obtain an upper bound for $W_m$.  According to the definition of
$M_{d,s}$,  the matrix $H\in M_{d,s}$ is in the form of
\[
H=\frac{\vh\vh^*+\vh\vx^*+\vx\vh^*}{\normf{\vh\vh^*+\vh\vx^*+\vx\vh^*}}
\]
 with $\vh/\norm{\vh}\in K_{d,s}$ and $\vx\in K_{d,s}$. Recall  that $A=\frac{1}{\sqrt{m}} \sum_{j=1}^m \epsilon_j \va_j\va_j^*$. This immediately leads to
\begin{eqnarray*}
 \nj{H,A} &= & \frac{1}{\normf{\vh\vh^*+\vh\vx^*+\vx\vh^*}} \cdot \frac{1}{\sqrt{m}} \sum_{j=1}^m \xkh{\epsilon_j \va_j^* \vh \vh^* \va_j  + \epsilon_j \va_j^* \vh \vx^* \va_j +  \epsilon_j \va_j^* \vx \vh^* \va_j}.
\end{eqnarray*}
According to the result (i),  there exists a sufficiently small constant $c_0>0$ such
that the following holds with probability at least $1-4\exp(-c m)$
\begin{equation} \label{eq:HAindetmine}
 \nj{H,A}\le  c_0 \sqrt{m} \cdot \frac{\norm{\vh}^2+2\norm{\vh}}{\normf{\vh\vh^*+\vh\vx^*+\vx\vh^*}},
\end{equation}
provided $m\gtrsim s \log (ed/s)$.  On the other hand, $H\in M_{d,s}$ implies
  \begin{equation} \label{eq:h2norm}
   \norm{\vh}\le 2\normf{\vh\vh^*+\vh\vx^*+\vx\vh^*}.
   \end{equation}
  We next show $ \norm{\vh}^2\le 5\normf{\vh\vh^*+\vh\vx^*+\vx\vh^*}$. Indeed,  by triangle inequality, we have
 \begin{equation}\label{eq:h2norm1}
\normf{\vh\vh^*+\vh\vx^*+\vx\vh^*} \ge \normf{\vh\vh^*}-\normf{\vh\vx^*+\vx\vh^*} \ge \norm{\vh}^2-2\norm{\vh}.
\end{equation}
 Combining (\ref{eq:h2norm}) and (\ref{eq:h2norm1}), we have
\begin{equation} \label{eq:hsquestim}
\norm{\vh}^2 \le \normf{\vh\vh^*+\vh\vx^*+\vx\vh^*}+2\norm{\vh} \le 5\normf{\vh\vh^*+\vh\vx^*+\vx\vh^*}.
\end{equation}
Putting (\ref{eq:h2norm}) and (\ref{eq:hsquestim}) into (\ref{eq:HAindetmine}), we
obtain that
\begin{equation} \label{eq:wmsparse}
W_m= \E\sup_{H\in M_{d,s}} \nj{H,A} \le 9c_0 \sqrt{m}.
\end{equation}
 Choosing $t=c\sqrt{2m}$ for a sufficiently small positive constant $c$ and putting
(\ref{eq:qxisparse}) and (\ref{eq:wmsparse}) into (\ref{eq:lowboundhsparse}), we
arrive at the conclusion.
\end{proof}

Based on the above lemmas, we are now ready to present the proof of Theorem \ref{result sparse}.

{\noindent\it Proof of Theorem \ref{result sparse}.}~ Without loss of generality, we
assume $\norm{\vx_0}=1$ (the general case can be obtained via a simple rescaling) and
$\nj{\x,\vx_0}\ge 0$ (Otherwise, we can choose $e^{i\theta} \vx_0$ for an appropriate
$\theta$).   Set $\vh:=\x-\vx_0$. We first show that $\normone{\vh} \le
2\sqrt{s}\norm{\vh}$. Indeed, let $S:=\supp(\vx_0)$. Then we have
\begin{equation*}
  \normone{\x}=\normone{\vx_0+\vh}=\normone{\vx_0+\vh_{S}}+\normone{\vh_{S^c}}\ge \normone{\vx_0}-\normone{\vh_{S}}+\normone{\vh_{S^c}}.
\end{equation*}
Here $\vh_{S}$ denotes the restriction of the vector $\vh$ onto the set of
coordinates $S$. Then the constrain condition $\normone{\x}\le R:=\normone{\vx_0}$
implies that $\normone{\vh_{S^c}}\le \normone{\vh_{S}}$. Using H\"older inequality,
we have
\begin{equation*}
  \normone{\vh}\,=\,\normone{\vh_{S}}+\normone{\vh_{S^c}}\, \le\, 2\normone{\vh_{S}}\,\le\, 2\sqrt{s}\norm{\vh}.
\end{equation*}
Set $H=\x\x^*- \vx_0 \vx_0^*$. It is straightforward to check that
\[
H=\vh\vh^*+\vh\vx_0^*+\vx_0\vh^*.
\]
From the claim (\ref{cla:uv}), we know $\norm{\vh}\le 2\normf{H}$.  Recall that
$\normone{\vh} \le 2\sqrt{s}\norm{\vh}$  and $\vx_0\in K_{d,s}$. It implies that
$H/\normf{H} \in M_{d,4s}$, where the set $M_{d,s}$ is defined in Lemma
\ref{le:cruciallem}.

 Since $\x$ is the global solution to (\ref{quartic model sparse}) and $\vx_0$ is a
feasible point, we have
\[
 \norm{ \A(\x\x^*)  -\vb} \le \norm{ \A(\vx_0\vx_0^*)  -\vb}
\]
which implies
\begin{equation} \label{eq:keypointsparse}
\norm{\A(H)-\eta}\le \norm{\eta}.
\end{equation}
Noting that $H/\normf{H} \in M_{d,4s}$, by  Lemma \ref{le:cruciallem},  we obtain
\begin{equation}\label{eq:lowressparse}
\sqrt{m}\normf{H} \lesssim \norm{\A(H)} \le \norm{\A(H)-\eta}+ \norm{\eta}\le  2\norm{\eta}
\end{equation}
with probability at least $1-4\exp(-c m)$, provided $m\gtrsim s \log (ed/s)$.
Thus,  (\ref{eq:lowressparse}) gives
\begin{equation} \label{eq:sparseH}
\normf{H}=\normf{\x\x^*- \vx_0 \vx_0^*}\lesssim \frac{\norm{\eta}}{\sqrt{m}},
\end{equation}
which implies
\[
 \min_{\theta\in [0,2\pi]} \norm{\x-e^{i\theta}\vx_0} \le \frac{2}{\norm{\vx_0}}\normf{\x\x^*-\vx_0\vx_0^*}\lesssim   \frac{ \norm{\eta}}{\norm{\vx_0}\sqrt{m}}.
\]
Here, we use  (\ref{cla:uv}).
 Based on (\ref{eq:sparseH}), similar to the proof of Theorem
\ref{result1}, we have
\[
 \min_{\theta\in [0,2\pi)} \norm{\x-e^{i\theta}\vx_0} \lesssim
  \|\vx_0\|_2+\frac{\sqrt{\|\eta\|_2}}{m^{1/4}}.
\]
It means that
\[
\min_{\theta\in [0,2\pi)} \norm{\x-e^{i\theta}\vx_0}\lesssim \min
\left\{ \|\vx_0\|_2+\frac{\sqrt{\|\eta\|_2}}{m^{1/4}}, \frac{ \norm{\eta}}{\norm{\vx_0}\sqrt{m}}\right\}.
\]
Finally, note that if $\|\vx_0\|_2\geq \frac{\sqrt{5}-1}{2}\cdot \frac{\sqrt{\|\eta\|_2}}{m^{1/4}}$ then
\[
\min\left\{\|\vx_0\|_2+\frac{\sqrt{\|\eta\|_2}}{{m}^{1/4}},
\frac{\|\eta\|_2}{\norm{\vx_0}\cdot \sqrt{m}}\right\}=\frac{\|\eta\|_2}{\norm{\vx_0}\cdot \sqrt{m}}
\]
and if $\|\vx_0\|_2< \frac{\sqrt{5}-1}{2}\cdot \frac{\sqrt{\|\eta\|_2}}{m^{1/4}}$ then
\[
\min\left\{\|\vx_0\|_2+\frac{\sqrt{\|\eta\|_2}}{{m}^{1/4}},
\frac{\|\eta\|_2}{\norm{\vx_0}\cdot
\sqrt{m}}\right\}=\|\vx_0\|_2+\frac{\sqrt{\|\eta\|_2}}{{m}^{1/4}}\leq
\frac{\sqrt{5}+1}{2}\cdot \frac{\sqrt{\|\eta\|_2}}{{m}^{1/4}}.
\]
We obtain the conclusion that
\[
\min_{\theta\in [0,2\pi)} \norm{\x-e^{i\theta}\vx_0}\lesssim \min
\left\{ \|\vx_0\|_2+\frac{\sqrt{\|\eta\|_2}}{m^{1/4}}, \frac{ \norm{\eta}}{\norm{\vx_0}\sqrt{m}}\right\}
\lesssim \min
\left\{ \frac{\sqrt{\|\eta\|_2}}{m^{1/4}}, \frac{ \norm{\eta}}{\norm{\vx_0}\sqrt{m}}\right\}.
\]
\qed

\section{Discussions}
This paper considers the  performance of the intensity-based estimators for phase
retrieval and its sparse version. The upper and lower bounds are obtained under
complex Gaussian random measurements.


There are some interesting problems for future research.
First, in the presence of noises, many numerical experiments show that gradient descent algorithms can solve
estimators (\ref{eq:mod1}) and (\ref{eq:mod2sparse}), however, it is of practical
interest to provide some theoretical guarantees for it.
 Second, a more practical scenario is the case where the measurements are Fourier
vectors. Since there is much less or even no randomness to be exploited  in this
scenario, we conjecture the estimation error would be no less than the lower bound
given in this paper, namely, $O(\norm{\eta}/\sqrt{m})$. To establish the precise
upper and lower bounds for Fourier measurements  is the future work.

\vspace{0.5cm}


%
%

 \appendix
 \renewcommand{\appendixname}{Appendix~\Alph{section}}

\section{Proof of Lemma \ref{le:upow4}}

The goal of this section is to prove Lemma \ref{le:upow4}. Before continuing, we introduce some lemmas.
The following result is a complex version of Lemma 5.8 in \cite{TWF}
and the proof is the same as that of Lemma 5.8 in \cite{TWF}.

\begin{lemma}  \label{le:indi}
Let $\va_j \in \C^d, j=1,\ldots,m,$ be i.i.d complex Gaussian random vectors. For any
$\epsilon>0$, there exist some universal constants $c_0,c_1,C>0$ such that
\[
\frac 1m \sum_{j=1}^m \1_{\dkh{\abs{\va_j^* \vz} \ge \gamma \norms{\vz}}} \le \frac{1}{0.49\gamma} \exp(-0.485 \gamma^2) +\frac{\epsilon}{\gamma^2} \quad \text{ for all } \; \vz \in \C^d\; \backslash \dkh{\bm{0}},  \gamma\geq 2
\]
holds with probability at least $1-C\exp(-c_0 \epsilon^2 m)$, provided $m\ge
c_1\epsilon^{-2} \log \epsilon^{-1}d$.

\end{lemma}

\begin{lemma} \label{le:trouble}
Suppose that $\va_j \in \C^d,j=1,\ldots,m $, are i.i.d. complex Gaussian random vectors. For
any $\epsilon \in (0,1)$, if $m\ge c(\epsilon) d \log d$ then the following holds
with probability at least $1-c'_a m^{-1}-c'_b \exp(-c'_c(\epsilon) m /\log m)$:
\[
\frac 1m \sum_{j=1}^m |\va_j^* \vz|^2 \Re(\vx_0^* \va_j \va_j^* \vz) \le 2\Re(\vx_0^* \vz) + \epsilon+2 \epsilon \xkh{\frac 1m \sum_{j=1}^m \abs{\va_j^* \vz}^4}^{\frac 34} \quad \mbox{for all} \; \vz\in \mathbb{S}_{\C}^{d-1},
\]
 where  $~\mathbb{S}_{\C}^{d-1}:=\{\vx\in \C^d: \|\vx\|=1\}$, $c'_a$, $c'_b$
 are positive absolute constants and $c(\epsilon)$, $c'_c(\epsilon)$ are positive constants
 depending on $\epsilon$.
\end{lemma}
\begin{proof}
Suppose that $\phi\in C_c^{\infty}(\mathbb R)$ is a Lipschitz continuous function
satisfying $0\le \phi(x)\le 1$ for all $x\in \mathbb R$. We furthermore require
$\phi(x)=1$ for $|x|\le 1$ and $\phi(x)=0$ for $|x|\ge 2$. For any   $\beta>0$, we
have
\begin{equation} \label{eq1}
\frac 1m \sum_{j=1}^m |\va_j^* \vz|^2 \Re(\vx_0^* \va_j \va_j^* \vz)= T+r
\end{equation}
where
\begin{equation*}
\begin{aligned}
 T&:= \frac 1m \sum_{j=1}^m |\va_j^* \vz|^2 \Re(\vx_0^* \va_j \va_j^* \vz) \phi\xkh{\frac {|\va_j^* \vz|}
 {\beta}},\\
r&:=\frac 1m \sum_{j=1}^m |\va_j^* \vz|^2 \Re(\vx_0^* \va_j \va_j^* \vz)
 \xkh{1- \phi\xkh{\frac {|\va_j^* \vz|}{\beta}}} \le \frac 1m \sum_{j=1}^m \abs{\va_j^* \vz}^3  \abs{\va_j^* \vx_0} \1_{\dkh{|\va_j^* \vz| \ge \beta}}.
 \end{aligned}
\end{equation*}
We claim that  for any $0<\epsilon<1$ there exists a  sufficiently large $\beta>1$
such that   if $m\ge c(\epsilon) d \log d$ then the following holds with probability
at least $1-c'_a m^{-1}-c'_b \exp(-c'_c(\epsilon) m /\log m)$:
\begin{equation}\label{eq:claim0}
T\le 2\Re(\vx_0^* \vz) + \epsilon,\quad
r \le 2 \epsilon \xkh{\frac 1m \sum_{j=1}^m \abs{\va_j^* \vz}^4}^{\frac 34}\quad \mbox{for all} \; \vz\,\, \in \mathbb{S}_{\C}^{d-1}.
\end{equation}
 Here
 $c(\epsilon), c'_c(\epsilon)  $ are constants depending on $\epsilon$ and $c_a', c_b'$ are positive absolute constants.
Substituting (\ref{eq:claim0}) into \eqref{eq1}, we obtain the conclusion that
 with probability at least $1-c'_a m^{-1}-c'_b \exp(-c'_c(\epsilon) m /\log m)$, it holds
\[
\frac 1m \sum_{j=1}^m |\va_j^* \vz|^2 \Re(\vx_0^* \va_j \va_j^* \vz) \le 2\Re(\vx_0^* \vz) + \epsilon+2 \epsilon \xkh{\frac 1m \sum_{j=1}^m \abs{\va_j^* \vz}^4}^{\frac 34} \quad \mbox{for all} \quad \vz\in \mathbb{S}_{\C}^{d-1}
\]
provided $m\ge c(\epsilon) d \log d$.

It remains to prove (\ref{eq:claim0}). We first show $T\le 2\Re(\vx_0^* \vz) +
\epsilon$. Due to the cut-off $\phi\xkh{\frac {|\va_j^* \vz|}{\beta}}$, the terms
$|\va_j^* \vz|^2 \Re(\vx_0^* \va_j \va_j^* \vz) \phi\xkh{\frac {|\va_j^* \vz|}
{\beta}}$ are independent sub-gaussian random variables with the sub-gaussian norm
$O(\beta^3)$. According to Hoeffding's inequality, we obtain that the following holds
with probability at least $1-2\exp(-c(\beta)\epsilon^2 m)$
\begin{equation}\label{eq:gudvz0}
\frac 1m \sum_{j=1}^m |\va_j^* \vz|^2 \Re(\vx_0^* \va_j \va_j^* \vz)
\phi\xkh{\frac {|\va_j^* \vz|}{\beta}} \le
2\Re(\vx_0^* \vz)+\frac {\epsilon} 3,
\end{equation}
 where $c(\beta)>0$ is a constant depending on $\beta$.
 Here we use the fact 
 \[
 \E\xkh{|\va_j^* \vz|^2 \Re(\vx_0^* \va_j \va_j^* \vz)\phi\xkh{\frac {|\va_j^* \vz|} {\beta}}
 }\leq \E\xkh{|\va_j^* \vz|^2 \Re(\vx_0^* \va_j \va_j^* \vz)} + \frac {\epsilon}6=2\Re(\vx_0^* \vz)+\frac {\epsilon} 6
\]
for some sufficiently large $\beta$ depending only on $\epsilon$.
 We next show that (\ref{eq:gudvz0}) holds
for all unit vectors $\vz\in \C^d$, for which we adopt a basic version of a
$\delta$-net argument. We assume that $\mathcal{N}$ is a $\delta$-net of the unit
complex sphere in $\Cd$  and hence the covering number $\# \mathcal{N}\le
(1+\frac{2}{\delta})^{2d}$. For any $\vz'\in  \mathbb{S}_{\C}^{d-1}$, there exists a
$\vz \in \mathcal{N}$ such that $\norm{\vz'-\vz}\le \delta$. Noting $f(\tau):=\tau^2
\phi(\tau/\beta)$ is a bounded function with Lipschitz constant $O(\beta)$, we
obtain  that when $m\gtrsim d \log d$, with probability at least $1-c_a m^{-1}-c_b
\exp(-c_c m /\log m)$, it  holds
\begin{equation}\label{eq:gudvz1}
\begin{aligned}
&\Big| \frac 1m \sum_{j=1}^m |\va_j^* \vz'|^2 \Re(\vx_0^* \va_j \va_j^* \vz') \phi\xkh{\frac {|\va_j^* \vz'|}{\beta}} - \frac 1m \sum_{j=1}^m |\va_j^* \vz|^2 \Re(\vx_0^* \va_j \va_j^* \vz) \phi\xkh{\frac {|\va_j^* \vz|} {\beta}}  \Big| \\
& \le  \frac 1m \sum_{j=1}^m |\va_j^* \vz'|^2 \phi\xkh{\frac {|\va_j^* \vz'|} {\beta}}  \abs{\vx_0^* \va_j \va_j^* \vz' - \vx_0^* \va_j \va_j^* \vz}  \\
&\quad +  \frac 1m \sum_{j=1}^m |\va_j^* \vz| |\va_j^* \vx_0 | \Big | |\va_j^* \vz'|^2 \phi\xkh{\frac {|\va_j^* \vz'|} {\beta}} - |\va_j^* \vz|^2 \phi\xkh{\frac {|\va_j^* \vz|} {\beta}} \Big|  \\
&\lesssim   \frac {\beta^2} m \sum_{j=1}^m \abs{\va_j^* \vx_0} \abs{\va_j^* \vz'-\va_j^* \vz} +  \frac {\beta} m \sum_{j=1}^m  |\va_j^* \vz| \abs{\va_j^* \vx_0} \abs{\va_j^* \vz'-\va_j^* \vz}  \nonumber\\
&\le \frac {\beta^2} m \norms{A\vx_0} \norms{A(\vz'-\vz)} +\beta  \sqrt{\frac 1m \sum_{j=1}^m  |\va_j^* \vz|^2 \abs{\va_j^* \vx_0}^2 } \cdot \norms{\frac 1 {\sqrt{m}}A(\vz'-\vz)}    \nonumber \\
&\le 2 \beta^2 \norms{\vz'-\vz} +3 \beta \norms{\vz'-\vz}
\le 5 \beta^2 \delta ,
\end{aligned}
\end{equation}
where the fourth inequality follows from Lemma \ref{le:sunju}  and the fact that
$\frac 1 {\sqrt{m}} \norms{A} \le \sqrt{2}$ with probability at least $1-2\exp(-c m)$
provided $m\gtrsim d$. Here, the matrix $A:=[\va_1,\ldots,\va_m]^*$ and $c_a, c_b,
c_c$ are absolute constants. Taking $\delta=\epsilon/(15 \beta^2)$, we use
(\ref{eq:gudvz0}) and (\ref{eq:gudvz1}) to obtain that if $m\ge c(\epsilon,\beta) d
\log d$ then with probability at least $1-c_a m^{-1}-c'_b \exp(-c'(\beta) \epsilon^2
m /\log m)$ it holds
\[
\frac 1m \sum_{j=1}^m |\va_j^* \vz|^2 \Re(\vx_0^* \va_j \va_j^* \vz) \phi\xkh{\frac {|\va_j^* \vz|} {\beta}} \le
2\Re(\vx_0^* \vz)+\epsilon
\]
for all  $\vz \in \mathbb{S}_{\C}^{d-1}$, where $c(\epsilon,\beta)$ is a positive
constant depending on $\epsilon, \beta$ and $c'(\beta)$ is a positive constant
depending on $\beta$.

We next show that $r \le 2 \epsilon \xkh{\frac 1m \sum_{j=1}^m \abs{\va_j^*
\vz}^4}^{\frac 34}$.
 By Lemma
\ref{le:indi}, for any $\epsilon>0$ there exists a sufficiently large $\beta$ such
that  if $m\ge c_1(\epsilon) d$ then with probability at least $1-C\exp(-c_0(
\epsilon) m)- c_2m^{-1}$ it holds
\begin{equation}\label{eq:rest}
\begin{aligned}
r &\le  \frac 1m \sum_{j=1}^m \abs{\va_j^* \vz}^3  \abs{\va_j^* \vx_0} \1_{\dkh{|\va_j^* \vz| \ge \beta}} \notag\\
&\le  \xkh{\frac 1m \sum_{j=1}^m \abs{\va_j^* \vz}^4}^{\frac 34} \xkh{\frac 1m \sum_{j=1}^m \abs{\va_j^* \vx_0}^8}^{\frac 18}  \xkh{\frac 1m \sum_{j=1}^m \1_{\dkh{|\va_j^* \vz| \ge \beta}}}^{\frac 18}  \notag \\
&\le  2 \epsilon \xkh{\frac 1m \sum_{j=1}^m \abs{\va_j^* \vz}^4}^{\frac 34} ,
\end{aligned}
\end{equation}
 where we use the Chebyshev's inequality in the last line to deduce that with probability at least $1-c_2m^{-1}$,
\[
\frac 1m \sum_{j=1}^m \abs{\va_j(1)}^8 \le 25.
\]
Here, $c_0(\epsilon)$ and $c_1(\epsilon)$ are constants depending on $\epsilon$,  and
$C, c_2$ are absolute constants.
\end{proof}

We are now ready to prove Lemma \ref{le:upow4}.
%
\begin{proof}[Proof of Lemma \ref{le:upow4}]
Without loss of generality, we assume $\norm{\vx_0}=1$ (the general case can be
obtained via a simple rescaling) and $\nj{\x,\vx_0}\ge 0$ (otherwise, we can choose
$e^{i\theta} \vx_0$ for an appropriate $\theta$).
 Recall that the loss function is
\[
f(\vz)=\sum_{j=1}^m \xkh{\abs{\nj{\va_j,\vz}}^2-b_j}^2.
\]
Since $\x $ is a global minimizer of $f(\vz)$ , we  have
\[
\nabla f(\x)=2\sum_{j=1}^m \xkh{\abs{\nj{\va_j,\x}}^2-b_j} \va_j\va_j^*\x=0.
\]
Let $\x:= R \z$, where $R\ge 0$ and $\norms{\z}=1$.   Recall that $b_j=|\va_j^*
\vx_0|^2+\eta_j, \;  j=1,\ldots,m$. Then $\nj{\nabla f(\x),\x}=0$ implies
\begin{equation}\label{eq:he1}
R^2 \sum_{j=1}^m \abs{\va_j^* \z}^4= \sum_{j=1}^m \abs{\va_j^* \z}^2 \abs{\va_j^* \vx_0}^2+   \sum_{j=1}^m \eta_j\abs{\va_j^* \z}^2.
\end{equation}
 Similarly, according to $\nj{\nabla f(\x),\vx_0}=0$  we have
\begin{equation}\label{eq:he2}
R^2 \sum_{j=1}^m \abs{\va_j^* \z}^2 \vx_0^* \va_j \va_j^* \z =\sum_{j=1}^m  \abs{\va_j^* \vx_0}^2 \vx_0^* \va_j \va_j^* \z + \sum_{j=1}^m \eta_j \vx_0^* \va_j \va_j^* \z.
\end{equation}
Combining (\ref{eq:he1}) and (\ref{eq:he2}), we obtain
\begin{equation}\label{eq:UV}
\begin{aligned}
U&:= { \frac 1m \sum_{j=1}^m \abs{\va_j^* \z}^4 \cdot \xkh{ \frac 1m\sum_{j=1}^m  \abs{\va_j^* \vx_0}^2 \Re(\vx_0^* \va_j \va_j^* \z) +  \frac 1m \sum_{j=1}^m \eta_j \Re( \vx_0^* \va_j \va_j^* \z)}}  \\
&= {\frac 1m\sum_{j=1}^m \abs{\va_j^* \z}^2 \Re(\vx_0^* \va_j \va_j^* \z) \cdot \xkh{\frac 1m\sum_{j=1}^m \abs{\va_j^* \z}^2 \abs{\va_j^* \vx_0}^2+  \frac 1m \sum_{j=1}^m \eta_j\abs{\va_j^* \z}^2}}=:V.
\end{aligned}
\end{equation}
Since $\nj{\x,\vx_0}\ge 0$, without loss of generality, we may assume $\z=s \vx_0+s_1 \vx_0^{\perp}$, where $\vx_0^{\perp} \in
\mathbb{S}_{\C}^{d-1}$ satisfies $\nj{\vx_0^{\perp}, \vx_0}=0$ and  $s,s_1$ are
positive real numbers obeying $s^2+s_1^2=1$.  A simple observation is that
$s:=\nj{\z,\vx_0} \in [0,1]$. We claim that for any $0<\epsilon<1$, when $m\ge
c(\epsilon) d\log m$, with probability at least $1-c_a \epsilon^{-2} m^{-1}-c_b
\exp(-c_c (\epsilon) d)$,  the followings hold:
\begin{equation}\label{eq:UU}
U \ge \frac 1m \sum_{j=1}^m \abs{\va_j^* \z}^4 \cdot \xkh{2s +  s \cdot  \frac 1 m \sum_{j=1}^m \eta_j - \lambda\epsilon }
\end{equation}
and
\begin{equation}\label{eq:VV}
V \le \xkh{2s + \epsilon+2 \epsilon \xkh{\frac 1m \sum_{j=1}^m \abs{\va_j^* \z}^4}^{\frac 34}} \cdot \xkh{1+s^2+ \frac 1 m \sum_{j=1}^m \eta_j + \lambda \epsilon },
\end{equation}
 where $\lambda$ is  a universal positive constant.
Here,  $c_a, c_b $ are absolute constants and $c(\epsilon), c_c(\epsilon)$ are
constants depending on $\epsilon$. Combining (\ref{eq:UV}), (\ref{eq:UU}) and
(\ref{eq:VV}), we obtain
\begin{equation}\label{eq:rela2}
\begin{aligned}
& \xkh{2s + \epsilon+2 \epsilon \xkh{\frac 1m \sum_{j=1}^m \abs{\va_j^* \z}^4}^{\frac 34}} \cdot \xkh{1+s^2+ \frac 1 m \sum_{j=1}^m \eta_j + \lambda \epsilon } \\
&\ge \frac 1m \sum_{j=1}^m \abs{\va_j^* \z}^4 \cdot \xkh{2s +  s\cdot  \frac 1 m \sum_{j=1}^m \eta_j - \lambda\epsilon }.
\end{aligned}
\end{equation}
According to Lemma \ref{le:sunju22},  when $m\ge c(\epsilon)d\log d$, with
probability at least $1-c_1\exp(-c_2(\epsilon) m)-c_3m^{-d}$,  it holds
\begin{equation} \label{eq:pow4low}
 \frac 1m \sum_{j=1}^m \abs{\va_j^* \z}^4\,\, \ge\,\, 2-\epsilon\,\, >\,\,1,
\end{equation}
where $c_2(\epsilon)$ is a constant depending on $\epsilon$ and $c_1, c_3$ are
absolute constants. Since $|\frac 1 m \sum_{j=1}^m \eta_j|$ is bounded, there exists a constant $C_0$ so that
\begin{equation}\label{eq:C0}
1+s^2+ \frac 1 m \sum_{j=1}^m \eta_j + \lambda \epsilon  \le C_0
\end{equation}
for some positive constant $C_0$.
 We can use (\ref{eq:pow4low}) and (\ref{eq:C0}) to obtain
\begin{equation}\label{eq:three3}
\begin{aligned}
2 \epsilon \xkh{\frac 1m \sum_{j=1}^m \abs{\va_j^* \z}^4}^{\frac 34}  \cdot \xkh{1+s^2+ \frac 1 m \sum_{j=1}^m \eta_j + \lambda \epsilon } &\le 2C_0 \epsilon \cdot \frac 1m \sum_{j=1}^m \abs{\va_j^* \z}^4,\\
\epsilon \cdot \xkh{1+s^2+ \frac 1 m \sum_{j=1}^m \eta_j + \lambda \epsilon } &\le C_0 \epsilon \cdot \frac 1m \sum_{j=1}^m \abs{\va_j^* \z}^4,\\
2s\lambda \epsilon &\le 2\lambda \epsilon \cdot \frac 1m \sum_{j=1}^m \abs{\va_j^* \z}^4.
\end{aligned}
\end{equation}
Substituting (\ref{eq:three3}) into  \eqref{eq:rela2}, we have
\begin{equation}\label{eq:ha1}
 2s\xkh{1+s^2+ \frac 1 m \sum_{j=1}^m \eta_j }\ge   \frac 1m \sum_{j=1}^m\abs{\va_j^* \z}^4  \cdot \xkh{2s +  s \cdot  \frac 1 m \sum_{j=1}^m \eta_j - C_1 \epsilon },
\end{equation}
where $C_1:=3\lambda+3C_0$ is  bounded.
  Assume that $c_0$ is a constant satisfying
$\frac{1}{m}\sum_{j=1}^m\eta_j\leq c_0$.
 Using \eqref{eq:pow4low} again, we have
  \begin{equation}\label{eq:ha2}
   \frac 2 m \sum_{j=1}^m \eta_j  \le \frac1m \sum_{j=1}^m\abs{\va_j^* \z}^4\cdot  \xkh{\frac 1 m \sum_{j=1}^m \eta_j  + c_0\epsilon}.
  \end{equation}
Combining (\ref{eq:ha1}) and (\ref{eq:ha2}), we have
\begin{equation} \label{eq:rala5}
2s(1+s^2) \ge   \frac 1m \sum_{j=1}^m\abs{\va_j^* \z}^4 \xkh{2s-C_2\epsilon}  \ge (2-\epsilon)\xkh{2s-C_2\epsilon},
\end{equation}
where $C_2:=3\lambda+3C_0+c_0$ is  bounded.  We claim that $s > \frac{\sqrt{5}}{5}$.
Recall that $s\le 1$. By taking $\epsilon>0$ sufficiently small, it then follows from
\eqref{eq:rala5} that $s$ must be sufficient close to $1$. Then \eqref{eq:rala5}
implies that, for any $\gamma>0$,
 the following holds  with probability at least $1-c''(\gamma) m^{-1}-c''' \exp(-c''''(\gamma) d)$
\[
 \frac 1m \sum_{j=1}^m\abs{\va_j^* \z}^4 \le 1+s^2 + \frac{C_3\epsilon}{s} \le 2+\gamma,
\]
provided $m\ge c'(\gamma) d \log m$, where $c'(\gamma),c''(\gamma), c''''(\gamma)$
are constants depending $\gamma$, $C_3$ and
 $c'''$ are sufficiently large constant.

It remains to prove (\ref{eq:UU}), (\ref{eq:VV}) and $s\geq \frac{\sqrt{5}}{5}$.

We first show that  (\ref{eq:UU}) holds.  Lemma \ref{le:sunju} implies that for any
$0<\epsilon<1$,  when $m\ge c(\epsilon) d\log d$, with probability at least $1-c_a
\epsilon^{-2} m^{-1}-c'_b \exp(-c_c \epsilon^2 m /\log m)$,
\begin{equation}\label{eq:est1}
\frac 1m \sum_{j=1}^m  \abs{\va_j^* \vx_0}^2 \Re(\vx_0^* \va_j \va_j^* \z) \ge 2s-\epsilon.
\end{equation}
Here, $c(\epsilon)$ is a constant depending on $\epsilon$ and $c_a, c'_b, c_c$ are absolute constants.
On the other hand, note that $\norms{\eta}\lesssim \sqrt{m}$ and $\norms{\eta}_{\infty} \lesssim \log m$.
Taking  $\delta=\epsilon$ in  Lemma \ref{le:etalowbound},  we obtain that the following holds with probability at least $1-2\exp(-c_c \epsilon^2 d)$:
\begin{eqnarray}
\frac 1m  \sum_{j=1}^m \eta_j \Re( \vx_0^* \va_j \va_j^* \z) &\ge& s \cdot  \frac 1 m \sum_{j=1}^m \eta_j - \rho\epsilon \cdot \xkh{\frac{\sqrt{d}}{m} \norms{\eta}+ \frac dm \norms{\eta}_{\infty}} \nonumber\\
& \ge&  s \cdot  \frac 1 m \sum_{j=1}^m \eta_j -C' \epsilon \label{eq:est2}
\end{eqnarray}
for some universal positive constant $C'$, provided $m \ge   C'' \rho d \log m $. Here, $C''$ is a universal constant and $\rho$ is a constant depending on $\epsilon$.
Combining (\ref{eq:est1}) and (\ref{eq:est2}), we arrive at (\ref{eq:UU}).

We next turn to (\ref{eq:VV}). By Lemma \ref{le:sunju},  for any $0<\epsilon<1$,
when $m\ge c(\epsilon) d\log d$, with probability at least $1-c_a \epsilon^{-2}
m^{-1}-c'_b \exp(-c_c \epsilon^2 m /\log m)$,
\begin{equation}\label{eq:VV1}
\frac 1m \sum_{j=1}^m \abs{\va_j^* \z}^2 \abs{\va_j^* \vx_0}^2 \le 1+s^2 +\epsilon.
\end{equation}
 Lemma \ref{le:etalowbound} implies that,  with probability at least $1-2\exp(-c_c \epsilon^2 d)$, we have
\begin{eqnarray}
\frac 1m \sum_{j=1}^m \eta_j\abs{\va_j^* \z}^2 &\le& \frac 1 m \sum_{j=1}^m \eta_j + \rho \epsilon \cdot \xkh{\frac{\sqrt{d}}{m} \norms{\eta}+ \frac dm \norms{\eta}_{\infty}} \nonumber\\
&\le &   \frac 1 m \sum_{j=1}^m \eta_j +C' \epsilon   \label{eq:VV2}
\end{eqnarray}
provided $m \ge   C'' \rho^{-1} d \log m $.
According to  Lemma \ref{le:trouble}, we obtain  that the following holds   with probability at least $1-c_a m^{-1}-c'_b \exp(-c'( \epsilon) m /\log m)$,
\begin{equation}\label{eq:VV3}
\frac 1m \sum_{j=1}^m |\va_j^* \z|^2 \Re(\vx_0^* \va_j \va_j^* \z) \le 2s + \epsilon+2 \epsilon \xkh{\frac 1m \sum_{j=1}^m \abs{\va_j^* \z}^4}^{\frac 34}
\end{equation}
provided $m\ge c(\epsilon) d \log d$.
Here $c'(\epsilon)$ is a positive constant depending on $\epsilon$. Combining (\ref{eq:VV1}), (\ref{eq:VV2}) and (\ref{eq:VV3}), we obtain
(\ref{eq:VV}).

We still need to show that $s\geq \frac{\sqrt{5}}{5} $. From Theorem \ref{result1},
we know that for $m\gtrsim d$, with probability at least $1-\exp(-c m)$,
\[
\norms{\x-\vx_0} \le C\frac{\norms{\eta}}{\sqrt{m}}.
\]
It immediately gives
\begin{equation} \label{eq:finalone}
\nj{\z,\vx_0}\ge \frac12 \norms{\x} +\frac{1-C^2\norms{\eta}^2/m}{2\norms{\x}}.
\end{equation}
We claim that $s:=\nj{\z ,\vx_0}\ge \sqrt{5}/5$ where $\z:= \x/\norms{\x}$. Indeed, if $C\norms{\eta}/\sqrt{m} \ge 2/\sqrt{5}$ then \eqref{eq:finalone} gives
\[
s:=\nj{\z ,\vx_0}\ge \frac12 \norms{\x} \ge   \frac{C\norms{\eta}}{2\sqrt{m}} \ge \sqrt{5}/5 ,
\]
where we use the fact that $2C\norms{\eta}/\sqrt{m} \le \norms{\vx_0} =1$ and
$\norms{\x} \ge 1-C\norms{\eta}/\sqrt{m}$. On the other hand, if
$C\norms{\eta}/\sqrt{m} < 2/\sqrt{5}$ then \eqref{eq:finalone} implies
\[
s:=\nj{\z ,\vx_0}\ge \sqrt{1-C^2\norms{\eta}^2/m} >   \sqrt{5}/5,
\]
where we use the inequality $a+b\ge 2\sqrt{ab}$ for any positive real numbers $a,b$.
In summary, we obtain $s\geq \sqrt{5}/5$.
\end{proof}



\end{document}